\newif\ifdraft \drafttrue \draftfalse
\documentclass[a4paper, 12pt, twoside]{article}
\usepackage[utf8]{inputenc}
\usepackage{tgpagella}
\usepackage{eulervm}
\usepackage[OT1]{fontenc}
\usepackage[toc]{appendix}
\usepackage[%
  ,sorting=nyt%
  ,maxcitenames=3%
  ,mincitenames=2%  
  ,minalphanames=2%
  ,maxbibnames=100%
  ,backend=bibtex%
  ,style=alphabetic%
]{biblatex}

\bibliography{bibliography.bib}

\usepackage{amsmath, amsfonts, mathtools, amsthm, amssymb, stmaryrd, etoolbox, xargs}
\usepackage{calc}
\usepackage{stmaryrd}
\usepackage[head=6mm,margin=24mm, foot=16mm, top=30mm, bottom=36mm]{geometry}
\usepackage{xcolor}
\usepackage{graphicx}
\usepackage{subfig}
\usepackage[blocks]{authblk}
\usepackage{pdflscape}
\usepackage{afterpage}
\usepackage{caption}
\usepackage{booktabs, tabularx, colortbl}
\definecolor{tabularrowcolor}{gray}{0.9}
\definecolor{tabularrowtitle}{gray}{0.85}

\usepackage[inline]{enumitem}
\usepackage{etoolbox}
\setlist{parsep=0pt, itemsep={0.2\thmspace}, topsep={0.3\thmspace}}
\setlist[enumerate]{label=\enumstyle{\alph{*}}}
\newlist{enuminline}{enumerate*}{2}
\setlist[enuminline]{label={\hspace{.5ex plus .3ex minus .2ex}\rm\arabic{*})},itemjoin*={{ }and{ }}}
\setlist[trivlist]{parsep=0pt,partopsep=0pt,topsep=\thmspace}
\newlist{subthm}{enumerate}{3}
\setlist[subthm]{label=\enumstyle{\alph{*}},ref=\thethm\hspace*{1pt}\enumstyle{\alph{*}}}

\newcommand{\unchecked}{$\square$}

\newlist{checklist}{itemize}{5}
\setlist[checklist]{label=\unchecked}

\newcommand{\enumstyle}[1]{{\rm(#1)}}
\newcounter{subthmfalsecounter}

\usepackage{svg,hyperref}

\usepackage{fancyhdr}
\ifdraft
\fi

\ifdraft
  \usepackage[mathlines,right]{lineno}
  \linenumbers
  \usepackage[notref,notcite]{showkeys}
  \renewcommand*\showkeyslabelformat[1]{\scriptsize\normalfont\ttfamily #1}
\else
  \newcommand*\showkeyslabelformat[1]{}
\fi

\usepackage{tikz}
\usetikzlibrary{calc, arrows.meta, shapes}

\apptocmd\normalsize{%
 \abovedisplayskip=6pt plus 1pt minus 0.5pt
 \abovedisplayshortskip=0pt plus 2pt
 \belowdisplayskip=6pt plus 1pt minus 0.5pt
 \belowdisplayshortskip=0pt plus 2pt
}{}{}

\DeclareRobustCommand\sbseries{\fontseries{sb}\selectfont}
\newcommand{\defemph}[1]{\textcolor{blue!50!black}{\sbseries#1}}

\newcommand{\subwalkrel}[1]{%
    \ifstrequal{#1}{<}{\mathrel{\sqsubset}}{%
    \ifstrequal{#1}{>}{\mathrel{\sqsupset}}{%
    \ifdefequal{#1}{\leq}{\mathrel{\sqsubseteq}}{%
    \ifdefequal{#1}{\geq}{\mathrel{\sqsupseteq}}{%
        \subwalkrel\leq#1%   
    }}}}%
}%

\newcommand{\subwordrel}[1]{%
    \ifstrequal{#1}{<}{\mathrel{\vartriangleleft}}{%
    \ifstrequal{#1}{>}{\mathrel{\vartriangleright}}{%
    \ifdefequal{#1}{\leq}{\mathrel{\trianglelefteq}}{%
    \ifdefequal{#1}{\geq}{\mathrel{\trianglerighteq}}{%
        \subwordrel\leq#1%   
    }}}}%
}%

\newcommand{\mm}{\minmultsem}
\newcommand{\sms}{\semfont{SMS}}

\newcommand{\semfont}[1]{{\normalfont\textsf{#1}}}

\newcommand{\minmultsem}{\semfont{MM}}

\newcommandx{\matchspace}[2][1=R,2=D]{\textsc{Match}(#1,#2)}

\newcommand{\walks}{\mathsf{Walks}}
\newcommand{\edges}{\mathsf{EdgeBag}}
\newcommand{\rededges}{\mathsf{RedEdgeBag}}
\newcommand{\edgeset}{\mathsf{EdgeSet}}
\newcommand{\powerset}{\mathcal{P}}

\DeclareMathOperator{\len}{\mathsf{Len}}
\DeclareMathOperator{\src}{\mathsf{Src}}
\DeclareMathOperator{\tgt}{\mathsf{Tgt}}

\DeclareMathOperator{\lbl}{\mathsf{Lbl}}
\DeclareMathOperator{\matches}{\mathsf{Matches}}

\NewDocumentCommand{\intint}{d[] d(] d[) d()}{\IfNoValue{#1}{\llbracket#1\rrbracket}}

\newcommand{\arobase}{\raisebox{-0.2ex}{\fontfamily{ptm}\selectfont @}}
\newcommand{\mel}[2]{{\tt\href{mailto:#1@#2}{#1\arobase{}#2}}}

\makeatletter
\newcommand{\vm@date@separator}{\hspace*{0.15ex}\rule[0.4\vm@date@height]{1ex}{0.07\vm@date@height}\hspace*{0.15ex}}
\newcommand{\vmdatefont}[1]{#1}
\newcommand{\isotoday}{%
  \vmdatefont{
    \newdimen\vm@date@height%
    \setbox0=\hbox{0123456789}%
    \vm@date@height=\ht0 \advance\vm@date@height by -\dp0
    \the\year\vm@date@separator\two@digits{\month}\vm@date@separator\two@digits{\day}%
  }%
}

\newcommand{\orcidicon}[1][1.5ex]{\adjustbox{max height=#1, raise=\dimexpr ( #1 * 195 / 700 * -1)  \relax}{\includesvg{aux/orcid}}}

\newcommand{\orcidfull}[1]{\href{#1}{\orcidicon[1.1em]~\texttt{#1}}}

\newcommand{\Gc}{\mathcal{G}}

\makeatletter
\newcommand{\set}[2][]{\{#1{\kern1\nulldelimiterspace}#2{#1\kern1\nulldelimiterspace\}}}
\newcommand{\setst}{\@ifstar{\autosetst}{\paramsetst}}
\newcommand{\autosetst}[2]{\left\{\,#1\,\middle|\,#2\,\right\}}
\newcommand{\paramsetst}[3][]{#1\{\,#2\mathbin{#1|}#3{\,#1\}}}

\newcommand{\bag}{\@ifstar{\autobag}{\parambag}}
\newcommand{\autobag}[1]{\llbrace*#1\rrbrace*}
\newcommand{\parambag}[2][]{\llbrace[#1]#2\rrbrace[#1]}
\newcommand{\bagst}{\@ifstar{\autobagst}{\parambagst}}
\newcommand{\autobagst}[2]{\llbrace*\,#1\,\middle|\,#2\,\rrbrace*}
\newcommand{\parambagst}[3][]{\llbrace[#1]\,#2\mathbin{#1|}#3\,\rrbrace[#1]}

\newcommand{\llbrace}{\@ifstar{\leftllbrace}{\paramllbrace}}
\newcommand{\paramllbrace}[1][]{{#1\{\hspace*{-.25em}#1\{}}
\newcommand{\leftllbrace}{\left\lbrace\kern-3\nulldelimiterspace\middle\lbrace}
\newcommand{\rrbrace}{\@ifstar{\rightrrbrace}{\paramrrbrace}}
\newcommand{\paramrrbrace}[1][]{{#1\}}\hspace*{-.25em}{#1\}}}
\newcommand{\rightrrbrace}{\middle\rbrace\kern-3\nulldelimiterspace\right\rbrace}
\makeatother

\newcommand{\np}{\textbf{NP}}
\newcommand{\conp}{\textbf{coNP}}
\newcommand{\delayp}{\textbf{PolyDelay}}
\newcommand{\ptime}{\textbf{P}}
\newcommand{\nl}{\textbf{NL}}

\usepackage[framemethod=TikZ]{mdframed}

\usepackage{xpatch}
\makeatletter
\xpatchcmd{\endmdframed}
  {\aftergroup\endmdf@trivlist\color@endgroup}
  {\endmdf@trivlist\color@endgroup\@doendpe}
  {}{}
\makeatother
\newlength{\problemmargin}
\newlength{\problempadding}
\newlength{\problemsep}
\newlength{\problemtitlepadding}
\newcommand{\problembullet}{\textbullet~~}
\newcommand{\problemfont}[1]{\textsc{#1}}
\makeatletter
\newenvironment{problem}[1]{%
    \begin{mdframed}[%
        frametitle={\hspace*{\dimexpr\problemtitlepadding-\problempadding\relax}#1},
        frametitlefont=\large\sc,
        innerleftmargin=\problempadding,
        innerrightmargin=\problempadding,
        innerbottommargin=\problempadding,
        innertopmargin=\problempadding,
        frametitleaboveskip=\problemtitlepadding,
        frametitlebelowskip=\problemtitlepadding,
        usetwoside=false,
        leftmargin=\problemmargin,
        rightmargin=\problemmargin,
        skipabove={.75\baselineskip\@plus.2\baselineskip\@minus.2\baselineskip},
        skipbelow={.5\baselineskip\@plus.2\baselineskip\@minus.2\baselineskip},
        frametitlerule=true,
        frametitlebackgroundcolor=black!10,
        nobreak,
        roundcorner=0pt,
        linewidth=0.5pt,
        frametitlerulewidth=0.5pt,
        ]%
    \description[topsep=0pt,itemsep=\problemsep,labelindent=0cm,leftmargin=\widthof{\problembullet},font=~\textbullet~\normalfont\bf]%
}{%
    \enddescription\end{mdframed}%
}
\makeatother
\newlength{\thmspace}
\newtheoremstyle{vmstyle}%
    {\thmspace} % Space above
    {\thmspace} % Space below
    {\slshape\renewcommand{\emph}[1]{\defemph{#1}}} % Body font
    {} % Indent amount
    {\bfseries} % Theorem head font
    {.} % Punctuation after theorem head
    {2ex} % Space after theorem head
    {} % Theorem head spec (can be left empty, meaning `normal')
\theoremstyle{vmstyle}

\newcommand{\thmBlockFont}[1]{#1}

\newcounter{thm}

\newtheorem{corollary}[thm]{\thmBlockFont{Corollary}}
\newtheorem{conjecture}[thm]{\thmBlockFont{Conjecture}}
\newtheorem{definition}[thm]{\thmBlockFont{Definition}}

\newtheorem{lemma}[thm]{\thmBlockFont{Lemma}}
\newtheorem{proposition}[thm]{\thmBlockFont{Proposition}}

\newtheorem{remark}[thm]{\thmBlockFont{Remark}}

\newtheorem{theorem}[thm]{\thmBlockFont{Theorem}}

\newtheorem{falsepropositionX}{\thmBlockFont{Proposition}}

\newtheorem{falsetheoremX}{\thmBlockFont{Theorem}}

\newtheorem{falsecorollaryX}{\thmBlockFont{Corollary}}

\newtheorem{falselemmaX}{\thmBlockFont{Lemma}}

\newtheorem*{falsestatementX}{\thmBlockFont{\thestatement}}

\newlength{\minnoderadius}
\newlength{\shortenlength}
\newlength{\nodelinewidth}
\newlength{\arrowwidth}
\def\loopangle{24}
\tikzset{%
    bend angle=20,
    >={Stealth[width=6pt,length=6pt]},
    node/.style={circle, line width=\nodelinewidth, draw, black, inner sep=2pt, outer sep=.5*\nodelinewidth, minimum height=\minnoderadius, minimum width=\minnoderadius},
    vertex/.style={fill=black,inner sep=1.5pt, circle},
    state/.style={node},  
    run state/.style={draw,rounded rectangle},
    preedge/.style={-, draw, black, line width=1pt, rounded corners=5pt,pos=.4, shorten >=\shortenlength},
    edge/.style={preedge,->},
    redge/.style={preedge,<-},
    initialedge/.style={edge, shorten > =0pt, shorten < =0pt},
    finaledge/.style={redge, shorten > =0pt, shorten < =0pt},
    borderedge/.style={edge, -, color=white, line width=5pt, shorten >=\shortenlength-2pt, >={Stealth[width=12pt,length=12pt]}},
    vm loop/.style={edge, pos=.5, looseness = 8},    
    graph loop/.style={edge, pos=.5, looseness = 40, shorten >=4pt},    
    road/.style={edge, color=fred},
    hl/.style={edge, line width = 1.5mm, color=borange, shorten >=1pt},
    ferry/.style={edge,color=fpurple},
    gas/.style={edge,color=fblue},
    start/.style={edge,color=fblue},
    end/.style={edge,color=fblue},
    >=stealth,
    north west loop/.style={vm loop, in={\the\numexpr 135 + \loopangle\relax}, 
                                     out ={\the\numexpr 135 - \loopangle\relax}},
    north east loop/.style={vm loop, in={\the\numexpr 45 + \loopangle\relax}, 
                                     out ={\the\numexpr 45 - \loopangle\relax}},
    south west loop/.style={vm loop, in={\the\numexpr -135 + \loopangle\relax}, 
                                     out ={\the\numexpr -135 - \loopangle\relax}},
    south east loop/.style={vm loop, in={\the\numexpr -45 + \loopangle\relax}, 
                                     out ={\the\numexpr -45 - \loopangle\relax}},
    north loop/.style={vm loop, in={\the\numexpr 90 + \loopangle\relax}, 
                                out ={\the\numexpr 90 - \loopangle\relax}},
    south loop/.style={vm loop, in={\the\numexpr 270 - \loopangle\relax}, 
                                out ={\the\numexpr 270 + \loopangle\relax}},
    east loop/.style={vm loop, in={\the\numexpr 0 + \loopangle\relax}, 
                                out ={\the\numexpr 0 - \loopangle\relax}},
    west loop/.style={vm loop, in={\the\numexpr 180 - \loopangle\relax}, 
                                out ={\the\numexpr 180 + \loopangle\relax}},                
    node distance = \nodedist,
}

\tikzset{
        state/.style={draw,rectangle, rounded corners},
        out state/.style={state, dashed},
        transition/.style={draw,-stealth,very thick, shorten >=2pt},
        out transition/.style={transition, dashed},
        p1/.style={rouge},
        p2/.style={bleu},
        p3/.style={vert},
        gadget/.style={draw},
        }

\newlength{\nodedist}
\newlength{\initfinaldist}
\newcommand{\initialfinal}[2][0]{%
    \def\angleI{\the\numexpr #1 + 15 \relax}
    \def\angleII{\the\numexpr #1 - 15 \relax}
    \path (#2.\angleI) ++(#1:\initfinaldist) coordinate 
        (#2-initialfinal1-#1);
    \path[initialedge] (#2.\angleI) to         
        (#2-initialfinal1-#1);
    \path (#2.\angleII) ++(#1:\initfinaldist) coordinate     (#2-initialfinal2-#1);
    \path[finaledge] (#2.\angleII) to
      (#2-initialfinal2-#1);
}
\definecolor{vert}{rgb}{0,.55,0.20}
\definecolor{bleu}{rgb}{0,0,0.75}
\definecolor{rouge}{rgb}{.75,0,0}

\usepackage{multicol}
\DeclareMathSymbol{\in}{\mathbin}{symbols}{"32}
\newcommand{\noneqspacing}{
  \thickmuskip=5mu plus 3mu minus 2mu
  \medmuskip=3mu plus 2mu minus 2mu
}
\noneqspacing
\newcommand{\eqspacing}{
  \thickmuskip=12mu plus 3mu minus 3mu
  \medmuskip=5mu plus 3mu minus 2mu
}%
\AtBeginEnvironment{gather*}{\eqspacing}%
\AtBeginEnvironment{gather}{\eqspacing}%
\AtBeginEnvironment{equation}{\eqspacing}%
\AtBeginEnvironment{equation*}{\eqspacing}%
\AtBeginEnvironment{align}{\eqspacing}%
\AtBeginEnvironment{align*}{\eqspacing}%
\AtBeginEnvironment{multline}{\eqspacing}%
\AtBeginEnvironment{multline*}{\eqspacing}%

\usepackage{lipsum}

\newcommand{\thesubtitle}{%
}
\newcommand{\thetitle}{Matching walks that are minimal with respect to edge inclusion}

\title{\thetitle\ifdefempty{\thesubtitle}{}{\\[.25em] \Large \thesubtitle}}

\usepackage{adjustbox}

\author{Victor Marsault}
\affil{LIGM, Université Gustave Eiffel, CNRS\\
       \orcidfull{https://orcid.org/0000-0002-2325-6004}\\
       \mel{victor.marsault}{univ-eiffel.fr}
       }
       
\date{\isotoday}

\ifdraft
\fi

\begin{document}

\maketitle

\begin{abstract}
   In this paper we show that enumerating the set~$MM(G,R)$, defined below, cannot be done with polynomial-delay in its input~$G$ and~$R$, unless $P=NP$.
   $R$ is a regular expression over an alphabet~$\Sigma$, $G$ is directed graph labeled over~$\Sigma$, and~$MM(G,R)$ contains walks of~$G$.
   First, consider the set~$Match(G,R)$ containing all walks~$G$ labeled by a word (over $\Sigma$) that conforms to~$R$.
   In general,~$M(G,R)$ is infinite, and $MM(G,R)$ is the finite subset of~$Match(G,R)$ of the walks that are minimal according to a well-quasi-order~$\preceq$.
   It holds $w\prec w'$ if the multiset of edges appearing in $w$ is strictly included in the multiset of edges appearing in~$w'$.
   Remarkably, the set~$MM(G,R)$ always contains some walks that may be computed in polynomial time.  
   Hence, it is not the case that the preprocessing phase of any algorithm enumerating $MM(G,R)$ must solve an NP-hard problem.
\end{abstract}

\medskip

\small%
\paragraph{Keywords:} Enumeration complexity, Polynomial-delay, Directed Graphs, Pattern matching, Regular expressions, Regular Path Queries.

\normalsize \bigskip

%%%%%%%%%%%%%%%%%%%%%%%%%%%%%%%%%%%%%%%%%%%%%%%%%%%%%%%%%%%%%%%%%%%%%%%%%%%%%%%%
%%%%%%%%%%%%%%%%%%%%%%%%%%%%%%%%%%%%%%%%%%%%%%%%%%%%%%%%%%%%%%%%%%%%%%%%%%%%%%%%
%%%%%%%%%%%%%%%%%%%%%%%%%%%%%%%%%%%%%%%%%%%%%%%%%%%%%%%%%%%%%%%%%%%%%%%%%%%%%%%%
\section{Introduction}

Finding walks that match regular expressions is a crucial part of the modern query languages, especially
for property graph database management systems.  Indeed, Regular Path Queries (RPQs), which were introduced a few decades ago \cite{CruzMendelzonWood1987,ConsensMendelzon1990}, are now the core of the languages Cypher \cite{FrancisEtAl2018-sigmod,FrancisEtAl2018-arxiv}, GSQL~\cite{DeutschEtAl2019}, PGQL~\cite{PGQL1.4} and the recent ISO standard language GQL \cite{DeutschEtAl2022,GQL-ISO}.
They are also used in SQL as part as the PGQ recent extension \cite{DeutschEtAl2022,SQLPGQ-ISO} and in the RDF language SPARQL \cite{SPARQL1.1PP}.

\clearpage

An RPQ matches the walks conforming to an regular expression, and RPQs are well-behaved in the case where one is interested in the existence of a walk matching the regular expression.  However, users seems to need the matching walks in full, as indicated by the semantics of most aforementioned languages.
Since there are sometimes infinitely many matching walks, real query languages choose an \emph{RPQ semantics}, that is, a filter that selects a finite number of "good" matching walks to output to the user.
It turns out that none of the RPQ semantics currently used are ideal (see \cite[Intro.]{DavidFrancisMarsault2023} or \cite{MarsaultMeyer2024}), which gives some room
to look for new candidate RPQ semantics.  For instance, \cite{DavidFrancisMarsault2023} introduces and studies a new kind of \emph{Run-based} RPQ semantics.  In another paper, we describe a framework to compare RPQ semantics \cite{MarsaultMeyer2024}.
We also define quite a few candidate semantics there. Two of them are studied in \cite{Khichane2024} and the present note is about another one.

Namely, we study the computational property of Minimal-Multiset (\mm) semantics.
It is defined as follows:  $\mm$ keeps a matching walk~$w$ in the output if there is no other matching walk~$w'$ such that the edges of~$w$ are strictly included in the edges of~$w$ (here \emph{inclusion} stands for \emph{bag inclusion}, hence multiplicity matters).
We show that the problem of enumerating all matching walks kept by \mm{} is not \delayp{} unless $\ptime=\np$ (Theorem~\ref{th:enum-hardness}).
Moreover, deciding whether \mm{} keeps a given walk is \conp-complete (Theorem~\ref{th:memb-conp-compl-mm}).
Both result are in data complexity (see Remark~\ref{r:data-complexity} or \cite{AbiteboulHullVianu1995}).

It is also noteworthy that the aforementioned enumeration problem exhibits some unusual features.
% To this end, we
% On the other hand, this enumeration problem has a some particularities that make the proof interesting from the standpoint of enumeration complexity theory (See for instance~\cite{Strozecki2021}). 
Usually, proving that an enumeration problem is not in \delayp{} unless $\ptime\neq\np$
amounts to showing that one must solve an \np-hard problem before outputting the \emph{first} solution.
Typically, for every instance~$I$ of some \np-hard problem, one builds an instance for the enumeration problem for which the outputs are certificates for $I$. We refer to these as \emph{hard} outputs in the following.
In our case, this approach cannot work. Indeed, we will see that every instance of our problem always has \emph{easy} solutions, in the sense that they may be found in polynomial time.
These easy solutions must remain in polynomial numbers.
Moreover, the condition for a candidate hard solution to be actually output highly depends on the easy solutions.
See Remark~\ref{r:easy-walks} for details.
In some sense, the usual approaches shows that problems are \np-hard in preprocessing time, while our problem seems to be of lesser complexity.

% show that one needs to solve an \np-hard problem before outputting 
% Recall that an enumeration problem has several outputs to produce,
% that its \emph{preprocessing time} is the time required before producing the first output the first solution,
% and its \emph{delay} is the maximal time required between the production of two successive outputs.
% The complexity class \delayp{} contains the enumeration problems that may be solved by an algorithm
% with polynomial preprocessing time and delay.
% Usually, proving that an enumeration problem is not in \delayp{} (unless $\np\neq\conp$)
% amounts to showing that o . 
% In other word, one shows 

Our proof is built upon a folklore reduction used to show that the problem \problemfont{$k$-Edge-Disjoint-Paths} is \np-hard.
This reduction can be adapted to to show that evaluation of RPQs under trail semantics evaluation is \np-hard. 
We even reshow that latter result as a first step of our proof: it corresponds roughly to the \emph{red} part of the graph and statements up to Proposition \ref{p:3sat<=>trail}.
Note that it also classical that the problem \problemfont{Two-Edge-Disjoint-Path} is \np-hard (see for instance \cite[Sec.~10.2]{BangJensenGutin2009}); however, the proof being more involved, it is harder to adapt to our setting.

\paragraph{Outline}
After Preliminaries, Section \ref{s:problem} precisely defines the Minimal-Multiset semantics
and the problem we study.
Then, we state and prove our main result in Section~\ref{s:main-result}.
Finally, Section~\ref{s:side-result} shows a few side results, such as the
\conp-hardness of membership under \mm{}
and an analog of our main result for a Shortest-Minimal-Set (\sms) semantics, a variant of \mm{} also proposed in \cite{MarsaultMeyer2024}.

%%%%%%%%%%%%%%%%%%%%%%%%%%%%%%%%%%%%%%%%%%%%%%%%%%%%%%%%%%%%%%%%%%%%%%%%%%%%%%%%
%%%%%%%%%%%%%%%%%%%%%%%%%%%%%%%%%%%%%%%%%%%%%%%%%%%%%%%%%%%%%%%%%%%%%%%%%%%%%%%%
%%%%%%%%%%%%%%%%%%%%%%%%%%%%%%%%%%%%%%%%%%%%%%%%%%%%%%%%%%%%%%%%%%%%%%%%%%%%%%%%
\section{Preliminaries}
\label{s:prelim}

\paragraph{Labels, words and regular expressions}

Given a set of symbols~$\Sigma$, \defemph{words} over $\Sigma$ are the finite sequences
of letters $a_1\cdots a_n$, and we denote the \defemph{empty word} by $\varepsilon$.
A \defemph{language} over~$\Sigma$ is a subset of~$\Sigma^*$.
A \defemph{regular expression}~$R$ over a set~$\Sigma$ is a formula obtained inductively from the set of letters, one unary postfix function~${}^*$, and two binary functions $+$ and~$\cdot$, according to the following grammar.
\begin{equation} 
    R \mathrel{\coloneqq} \varepsilon \mid a \mid R^* \mid R\mathbin{\cdot}R \mid R+R   \quad\quad\text{where $a\in\Sigma$}
\end{equation}
We may omit the $\cdot$ operator, and use it mostly to help readability.
% When there is a parenthesizing ambiguity,~${}^*$ takes precedence over~$\cdot$, which takes precedence over~$+$.
% For instance,~$\varepsilon+a^*+bc^*=\varepsilon+(a^*)+(b(c^*))$.
A regular expression~$R$ over~$\Sigma$ \defemph{denotes} a language $L(R)\subseteq\Sigma^*$ inductively defined as usual:
    \begin{align*}
        L(\varepsilon) ={}& \{\varepsilon\}
        & L(R\cdot R') = {}& L(R)\cdot L(R')
        & L(R^*) ={}& L(R)^*
        \\
        \forall a\in\Sigma\quad L(a) ={}& \{a\} 
        & L(R+R') ={}& L(R)\cup L(R')
    \end{align*}
    
%%%%%%%%%%%%%%%%%%%%%%%%%%%%%%%%%%%%%%%%%%%%%%%%%%%%%%%%%%%%%%%%%%%%%%%%%%%%%%%%
%%%%%%%%%%%%%%%%%%%%%%%%%%%%%%%%%%%%%%%%%%%%%%%%%%%%%%%%%%%%%%%%%%%%%%%%%%%%%%%%
\paragraph{Graphs, walks and matches}
In this document, we consider (directed labeled) graph.
Formally, a graph is a triplet $\Gc=(V,L,E)$
where~$V$ is a finite set of \defemph{vertices}, $L$ is a finite set of \defemph{labels}
and~$E\subseteq V\times L \times V$ is a finite set of \defemph{edges}.
An element of~$E$ will be denoted by~$s\xrightarrow{a} t$ in order to help readability.

A walk in~$\Gc$ is an alternating sequence of vertices and labels consistent with the edges of $\Gc$.
More precisely, a \defemph{walk}~$w$ is a sequence~$(v_0,a_1,v_1,\ldots,a_k,v_k)$ where each $v_i$ belongs to $V$,
each $a_j$ belongs to $L$, and $(v_{i-1},a{i},v_{i})$ belongs to $E$ for every $i\in\{1,\ldots,k\}$.
Moreover, $k$ is called the \defemph{length} of~$w$, written $\len(w)$; $v_0$ is called the \defemph{source} of~$w$, written $\src(w)$; and $v_k$ is called the \defemph{target} of~$w$, written $\tgt(w)$.
As we do for edges, we generally use arrows to denote walks, as in $v_0\xrightarrow{a_1} v_1\xrightarrow{}\cdots \xrightarrow{a_k} v_k$.
The set of all walks in~$\Gc$ is denoted by $\walks(\Gc)$; note that vertices are walks of length $0$ and edges are walks of length $1$.
A \defemph{trail} is a walk with no repeated edge.

Two walks $w=(v_0\xrightarrow{a_1} \cdots \xrightarrow{a_k} v_k)$ and $w'=(v_0'\xrightarrow{a_1'} \cdots \xrightarrow{a_n'} v_n')$ \defemph{concatenate} if $v_k=v'_0$ and their \defemph{concatenation} is the walk $w\cdot w'=(v_0\xrightarrow{a_1} v_1\xrightarrow{}\cdots \xrightarrow{a_k} \underset{=v_0'}{v_k}\xrightarrow{a_1'} v_1'\xrightarrow{}\cdots \xrightarrow{a_n'} v_n')$.
Given two walks~$w,w'$, we say that~$w$ is a \defemph{factor} of~$w'$ if there exists two walks $w_p$ and $w_s$
such that $w'=w_p \cdot w\cdot w_s$.
In particular, a walk is a factor of itself since~$w_p$ and $w_s$ can have length $0$.

%%%%%%%%%%%%%%%%%%%%%%%%%%%%%%%%%%%%%%%%%%%%%%%%%%%%%%%%%%%%%%%%%%%%%%%%%%%%%%%%
%%%%%%%%%%%%%%%%%%%%%%%%%%%%%%%%%%%%%%%%%%%%%%%%%%%%%%%%%%%%%%%%%%%%%%%%%%%%%%%%
\paragraph{Enumeration complexity}

We generally use the enumeration complexity framework, which we sketch below.
See for instance \cite{Strozecki2021} for more formal definitions.
An enumeration problem~$E$ takes inputs~$I$ and outputs
a finite set $E(I)$ of \defemph{solutions}.
An enumeration algorithm~$A$ solves $E$ if it outputs each element in $E(I)$, in an order at the discretion of the algorithm.
We denote by $T_A(I,k)$ the number of steps taken by~$A$ with input $I$ before outputting the $k$-th solution and, writting~$\ell$ for the cardinal of $E(I)$, we denote by $T_A(I,\ell+1)$ the number of steps taken by~$A$ to stop.
%and we set $T_A(I,0)=0$.
The \defemph{preprocessing time} of~$A$ is the function mapping an instance~$I$ of~$E$ to $T_A(I,1)$
and its \defemph{delay} is the function mapping each instance~$I$ of $E$ to $max_{i\in\{1,..,\ell\}}\big(T_A(I,k+1)-T_A(I,k)\big)$, where~$\ell$ is the cardinal of~$E(I)$.
The algorithm $A$ is called \defemph{poly-delay} if its preprocessing time and its delay 
are bounded by polynomials in the size of~$I$.
% there is a polynomial~$P$ such that, for every~$k$, $(T_A(I,k+1)-T_A(I,k))<P(|I|,|e_k|)$
% where $|e_k|$ is the size of the $k$-th solution.

%%%%%%%%%%%%%%%%%%%%%%%%%%%%%%%%%%%%%%%%%%%%%%%%%%%%%%%%%%%%%%%%%%%%%%%%%%%%%%%%
%%%%%%%%%%%%%%%%%%%%%%%%%%%%%%%%%%%%%%%%%%%%%%%%%%%%%%%%%%%%%%%%%%%%%%%%%%%%%%%%
%%%%%%%%%%%%%%%%%%%%%%%%%%%%%%%%%%%%%%%%%%%%%%%%%%%%%%%%%%%%%%%%%%%%%%%%%%%%%%%%
\section{Considered problem}
\label{s:problem}

\begin{definition}
    Given a regular expression~$R$ and a graph~$\Gc$, a \defemph{match} to~$R$ in~$\Gc$
    is a walk~$w$ such that~$\lbl(w)\in L(R)$.
    We denote by $\matches(\Gc,R)$ the set of all matches, and given moreover two vertices $s,t\in V^2$, we denote by $\matches(\Gc,R,s,t)$ the set of all matches with source~$s$ and target~$t$.
\end{definition}

Due to cycles in the graph and Kleene stars in the expression, $\matches(\Gc,R)$ and $\matches(\Gc,R,s,t)$ is infinite in some cases.
Languages use \emph{RPQ semantics} (such as \emph{trail semantics} or \emph{shortest semantics}) to selects finitely many matches to output to the user (see~\cite{DavidFrancisMarsault2023,MarsaultMeyer2024}).
In this document, we consider a finite subset of matches defined by the function $\mm$ below.

\begin{definition}[Minimal-Multiset Semantics]
    \begin{subthm}
        \item \label{def:minimal-multiset}
        For every walk~$w$, we let~$\edges(w)$ denote the \textbf{multiset} of the edges appearing in~$w$.
        \item 
        We define the strict partial order relation~$\prec$ on walks as follows:~$w\prec w'$ if and only if $\edges(w)\subsetneq \edges(w')$.
        
        As usual, we let~$\preceq$ denote the reflexive extension of~$\prec$.
        \item We call \emph{minimal multiset semantics} the function~$\mm:\Gc\times R\to \powerset(\walks(\Gc))$ defined as follows.
        \begin{multline}\label{eq:mm}
            \mm(\Gc,R) = \bigcup_{\smash{(s,t)\in V_\Gc^{~2}}} \mm(\Gc,R,s,t) \\ 
            \text{where}\quad\mm(\Gc,R,s,t) = \min_\prec \matches(\Gc,R,s,t)
        \end{multline}
    \end{subthm}
\end{definition}

Note that there is no infinite descending chain~$w_0\succ w_1\succ \cdots \succ w_k\succ \cdots$, meaning that $\min_\prec$ is well-defined in Equation~\eqref{eq:mm}.
In fact, one may verify that $\preceq$ is a well-quasi-order over~$\walks(\Gc)$, from which the next statement follows.

\begin{lemma}\label{l:finite}
    For every graph~$\Gc$ and regular expression~$R$, then~$\mm(\Gc,R)$ is finite.
\end{lemma}

% \begin{remark}
% If we were using a $\edges(w)$ were defined as the \textbf{set} of the edges appearing in~$w$, then Lemma~\ref{l:finite} would not hold.
% \end{remark}

Checking nonemptyness of $\matches(\Gc,R,s,t)$ amounts to checking accessiblity in some kind of product graph of $\Gc\times R$, a problem known to be \nl-complete \cite{AroraBarak2009}.
Since $\mm(\Gc,R,s,t)$ is empty if and only if $\matches(\Gc,R,s,t)$ is, the next statement follows.

\begin{lemma}\label{lem:emptyness}
    The following problem is \nl-complete: given a graph~$\Gc$, a regular expression~$R$
    and two vertices, is $\mm(\Gc,R,s,t)$ nonempty ?
\end{lemma}

%%%%%%%%%%%%%%%%%%%%%%%%%%%%%%%%%%%%%%%%%%%%%%%%%%%%%%%%%%%%%%%%%%%%%%%%%%%%%%%%
%%%%%%%%%%%%%%%%%%%%%%%%%%%%%%%%%%%%%%%%%%%%%%%%%%%%%%%%%%%%%%%%%%%%%%%%%%%%%%%%
%%%%%%%%%%%%%%%%%%%%%%%%%%%%%%%%%%%%%%%%%%%%%%%%%%%%%%%%%%%%%%%%%%%%%%%%%%%%%%%%
\section{Main result}
\label{s:main-result}

The purpose of this document is to show that it is not possible to enumerate $\mm(\Gc,R,s,t)$ with a polynomial delay algorithm, as formally stated below.

\begin{problem}{Walk Enum Under \mm}
  \item[Input:] A regular expression~$R$, a graph~$\Gc$ and two vertices~$s,t$.
  \item[Output:] All walks in $\mm(R,\Gc,s,t)$.
\end{problem}

\begin{theorem}\label{th:enum-hardness}
        Unless $\ptime=\np$, there exists no polynomial delay algorithm to solve \problemfont{Walk Enum Under \mm}. It is already true for a fixed regular expression~$R$ with star-height~$1$.
\end{theorem}

The remainder of Section~\thesection{} is dedicated to proving Theorem~\thetheorem{}.
We give a scheme of the proof in Section~\ref{sec:scheme}, together with the notation
about the 3-SAT instance~$I$ we reduce from.
Afterwards, the fixed regular expression~$R$ is then given in Section~\ref{sec:regexp}.
We then present in section~\ref{sec:db} the graph~$\Gc_I$, built from the instance~$I$.
Finally, Section~\ref{sec:real-proof} states and shows the claims leading to Theorem~\thetheorem{}.

\begin{remark}\label{r:easy-walks}
    The idea underlying Lemma~\ref{lem:emptyness} also show that finding \textbf{some} walks in $\mm(\Gc,R,s,t)$ is easy.
    Indeed, the shortest walk in $\matches(\Gc,R,s,t)$ is minimal w.r.t.$\prec$, hence also belongs to $\mm(\Gc,R,s,t)$.
    Several other short walks are also likely to belong to $\mm(\Gc,R,s,t)$.
    These walks may be computed in polynomial time with a breadth-first search of $\Gc\times R$.

    The tension in the proof of Theorem~\ref{th:enum-hardness} comes from the existence 
    of these easy solutions.
    We need to keep their number small (in polynomial number) otherwise one may solve the 3-SAT instance~$I$ by executing one step of the computation before outputting each of the easy solutions.
    On the other hand, as explained later in Section~\ref{sec:scheme}, there will be exponentially many walks in $\matches(\Gc,R,s,t)$, one for each
    valuation of the SAT variables.
    Hence each of the exponentially many valuations that are not a witness to $I$ must be greater w.r.t.~$\prec$ than one of the polynomially many easy solutions.
\end{remark}

\begin{remark}\label{r:data-complexity}
In database theory, it is usual to consider \textit{data complexity}, that is the complexity when the query is considered to be constant (see for instance~\cite{AbiteboulHullVianu1995}).
In our case, the graph~$\Gc$ is the abstraction of the database and~$R$ is the abstraction of the query. Hence, Theorem~\thetheorem{}
implies that \problemfont{Walk Enum Under \mm} is not poly-delay in data complexity (unless $\ptime=\np$).
\end{remark}

%%%%%%%%%%%%%%%%%%%%%%%%%%%%%%%%%%%%%%%%%%%%%%%%%%%%%%%%%%%%%%%%%%%%%%%%%%%%%%%%
%%%%%%%%%%%%%%%%%%%%%%%%%%%%%%%%%%%%%%%%%%%%%%%%%%%%%%%%%%%%%%%%%%%%%%%%%%%%%%%%
\subsection{Proof outline}
\label{sec:scheme}

Although the term is not proper, our proof consists in a many-one reduction from 3-SAT.
% , in the sense that an instance of the \problemfont{Walk Enum Under \mm} is built from an instance of 3-SAT.
In the following, we fix a 3-SAT instance~$I$ with~$\ell$
clauses $C_1,\ldots,C_j$ and $k$ variables~$x_1,\ldots,x_k$.
We denote negated literal with a bar, as in~$\bar{x_i}$.
From~$I$, we will build a graph $\Gc_I$ with two vertices $Source$ and $Target$;
we will also partition the edge into three color (red, green and blue) to help readability.

The last input of \problemfont{Walk Enum Under \mm} is a fixed expression~$R$ that will be of the shape~$R=R_1\cdot R_3  + R_2$.
First, $R_1$ is such that it matches a trail in $\Gc_I$ if and only $I$ is satisfiable.
Moreover matches to $R_1$ contain only red edges.
Then, $R_2$ matches walks that are almost trails, in the sense that exactly one edge is used twice.
The property we want is: ($\ast$) for every match~$w_1$ to~$R_1$ that is not a trail, there exists a match~$w_2$ to~$R_2$
such that $w_2\prec w_1$ (that is, $\edges(w_2)\subsetneq\edges(w_1)$), in which case $w_1$ is not returned. Indeed, since a nontrail must use some edge twice, there is a match to~$R_2$ that uses that edge twice.
However the property ($\ast$) does not precisely hold because matches to $R_2$ use extra (blue) edges.
The expression $R_3$ solves this problem by matching exactly one walk which uses every extra edge in the graph (in fact every blue or green edge).

As explained in Remark~\ref{r:easy-walks}, \problemfont{Walk Enum.~Under \mm} always has easy solutions.
In our case, the matches to~$R_2$ are the easy solutions and we will see that there is a polynomial number of them.
After they are output, the next solution must be a match to $R_1\cdot R_3$ that is a trail, if it exists.
This provides an algorithm solving~$I$: An hard solution is output if and only if~$I$ is satisfiable.

%%%%%%%%%%%%%%%%%%%%%%%%%%%%%%%%%%%%%%%%%%%%%%%%%%%%%%%%%%%%%%%%%%%%%%%%%%%%%%%%
%%%%%%%%%%%%%%%%%%%%%%%%%%%%%%%%%%%%%%%%%%%%%%%%%%%%%%%%%%%%%%%%%%%%%%%%%%%%%%%%
\subsection{The regular expression~$R$}
\label{sec:regexp}
\newcommand{\R}{R_1R_3  + R_2}
\newcommand{\Ri}{0\cdot (1+2+313)^*\cdot 0}
\newcommand{\Rii}{0\cdot2^*\cdot 55 \cdot 1^* \cdot 4 \cdot 1^* \cdot 55 \cdot 2^* \cdot 0}
\newcommand{\Riii}{9\cdot (8 + 755 + 646 + 557)^*\cdot X}

    The expression~$R$, defined below, is over the alphabet~$\Sigma=\{0,\ldots, 9,X\}$, has three 
    part~$R_1$,~$R_2$ and~$R_3$ that will be used later on in the proof.
    \begin{equation}
        R = \R\quad\quad
        \text{where}\quad\left\{\begin{array}{r@{}l}
             R_1 ={}& \Ri\\
             R_2 ={}& \Rii\\
             R_3 ={}& \Riii
        \end{array}\right.
    \end{equation}
   
    \begin{remark}
    Note that $R$ is of star-height 1 and is independent of the SAT instance~$I$.
    \end{remark}

    \noindent We also split~$\Sigma$ into three: $\Sigma=\Sigma_{R}\uplus\Sigma_{B}\uplus\Sigma_{G}$ where
    \begin{itemize}
        \item $\Sigma_{R}=\{0,1,2,3\}$ contains the \emph{red} letters;
        \item $\Sigma_{B}=\{4,5\}$ contains the \emph{blue} letters;
        \item $\Sigma_{G}=\{6,7,8,9,X\}$ contains the \emph{green} letters.
    \end{itemize}
    
    \begin{remark}\label{rem:color}
        Expression~$R_1$ is red, that is $R_1$ is over $\Sigma_{R}$.
        Expression~$R_2$ is red and blue, that is $R_2$ is over $\Sigma_{R}\cup\Sigma_{B}$.
        Expression~$R_3$ is blue and green, that is $R_3$ is over $\Sigma_{B}\cup\Sigma_{G}$.
    \end{remark}

%%%%%%%%%%%%%%%%%%%%%%%%%%%%%%%%%%%%%%%%%%%%%%%%%%%%%%%%%%%%%%%%%%%%%%%%%%%%%%%%
%%%%%%%%%%%%%%%%%%%%%%%%%%%%%%%%%%%%%%%%%%%%%%%%%%%%%%%%%%%%%%%%%%%%%%%%%%%%%%%%
\subsection{Presentation of the gadgets and the graph~$\Gc_I$}
\label{sec:db}

Figures \ref{fig:gadget-start} to \ref{fig:gadget-glu} show the different gadgets
used in~$\Gc_I$.
They collectively define~$G_I$ entirely.
We use the following drawing convention.
\begin{itemize}
    \item The color of an edge corresponds to the color of the label held by the edge.
    \item Vertices drawn with a plain line are inside the gadget and vertices drawn with a dashed line are outside the gadget.
    \item Edges drawn with a plain line have both endpoints in the gadget.
    \item Edges drawn with a dotted or a dashed line have one of their endpoints outside of the gadget.
    \item Dashed edges are unconditional, while
    dotted edges (in Fig~\ref{fig:gadget-var} only) might or might not exist depending on the instance~$I$.
\end{itemize}

\medskip 

Aside from the glu gadget, described last, they are given roughly in the order a (red) walk matching~$R_1$ will traverse them.

 Figure~\ref{fig:gadget-start} gives the \defemph{start gadget}. 
        Aside from the vertex $Source$, it contains 2 vertex for each variable~$x_i$,
        that connects to the beginning of the positive and negative side of gadget $x_i$, respectively.
        Each walk matched by~$R_1$, $R_2$ or~$R$ will start in vertex $Source$ in this gadget.

\begin{figure}[hp]\centering
    \makebox[\linewidth][c]{\begin{tikzpicture}
    \path node[state] (start) {$Source$}
        ++(0:35mm) node[state] (start1) {$start_{x_1}$}
        ++(0:24mm) node[state] (startn1) {$start_{\overline{x_1}}$}
        % ++(0:24mm) node[state] (start2) {$start_2$}
        % ++(0:24mm) node[state] (startn2) {$\overline{start}_2$}
        ++(0:3cm) node[state] (startk) {$start_{x_k}$}
        ++(0:24mm) node[state] (startnk) {$start_{\overline{x_k}}$}
        ++(0:3cm) node[out state] (x0) {$x_0\rhd$}
        ;
        
    % \path[p1, out transition] (Cl) to node[above] {$2$} (start1);
    \path[p1, transition] (start) to node[above] {$0$} (start1);
    \path[p1, transition] (start1) to node[above] {$2$} (startn1);
    % \path[p1, transition] (startn1) to node[above] {$2$} (start2);
    % \path[p1, transition] (start2) to node[above] {$2$} (startn2);
    \path (startn1) to node{$\cdots$} (startk);
    \path[p1, transition] (startk) to node[above] {$2$} (startnk);
    \path[p1, out transition] (startnk) to node[above] {$2$} (x0);

    \path[draw] ($(start.north west)+(-5mm,5mm)$) -- ($(startnk.north east)+(5mm,5mm)$) -- ($(startnk.south east)+(5mm,-5mm)$) -- ($(start.south west)+(-5mm,-5mm)$) -- cycle;
    
    \path (start1) ++(90:3cm) node[out state] (SX1) {$Sx_1\rhd $};   
    \path[p2,out transition] (start1) to[bend right] node[right,midway] {$5$} (SX1);
    \path[p3,out transition] (SX1) to[bend right] node[left,midway] {$7$} (start1);

    \path (startn1) ++(90:3cm) node[out state] (SnX1) {$S{\bar x}_1\rhd $};   
    \path[p2,out transition] (startn1) to[bend right] node[right,midway] {$5$} (SnX1);
    \path[p3,out transition] (SnX1) to[bend right] node[left,midway] {$7$} (startn1);

    \path (startk) ++(90:3cm) node[out state] (SXk) {$Sx_k\rhd $};   
    \path[p2,out transition] (startk) to[bend right] node[right,midway] {$5$} (SXk);
    \path[p3,out transition] (SXk) to[bend right] node[left,midway] {$7$} (startk);

    \path (startnk) ++(90:3cm) node[out state] (SnXk) {$S{\bar x}_k\rhd $};   
    \path[p2,out transition] (startnk) to[bend right] node[right,midway] {$5$} (SnXk);
    \path[p3,out transition] (SnXk) to[bend right] node[left,midway] {$7$} (startnk);
    
\end{tikzpicture}}
    \caption{Start gadget}
    \label{fig:gadget-start}
\end{figure}
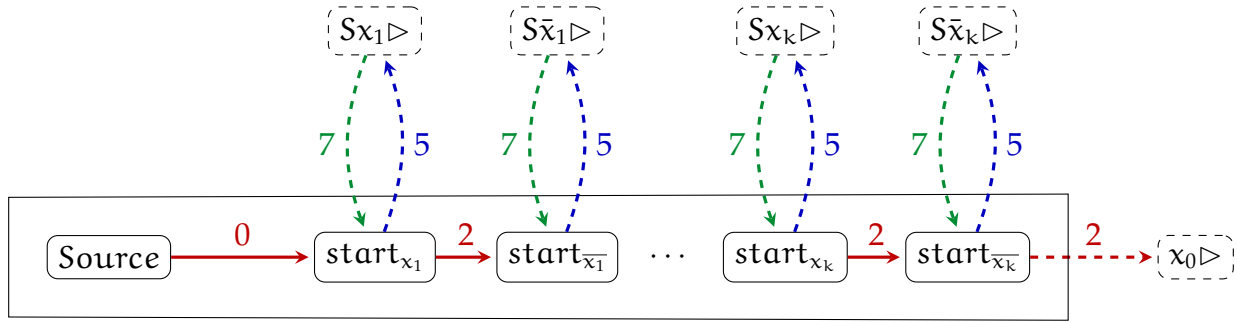
        
For each variable~$x_i$, $\Gc_I$ contains one \defemph{variable} gadget defined in Figure~\ref{fig:gadget-var}. It is formed of a \defemph{positive side} (top part: $\rhd Sx_i \rightarrow x^0_i\rightarrow \cdots \rightarrow x^\ell_i$),
    of a \defemph{negative side} (bottom part: $\rhd S{\bar x}_i \rightarrow {\bar x}^0_i\rightarrow \cdots \rightarrow {\bar x}^\ell_i$), and of two vertices $\rhd x_i$ and $x_i\rhd$.
    It is linked to the prior variable gadget (left) and the next variable gadget (right).
    Finally, two edges connecting $C_j$ to the positive (resp. negative) side of this gadget if $x_i$ (resp $\lnot x_i$) is a literal in $C_j$.

\begin{figure}[hpt]\centering
    \makebox[\linewidth][c]{\begin{tikzpicture}

    \node[state] (X0) {$\rhd x_i$};

    \path (X0) ++(2cm,12mm) node[state] (x00) {$x_i^0$} ++(0:2cm) node[state] (x01) {$x_i^1$} ++(0:2cm) node[state] (x02) {$x_i^2$} ++(0:3cm) node[state] (x0l) {$x_i^{\ell}$};
    
    \path (X0) ++(90:2cm) node[state] (SX0) {$\rhd Sx_i$};    
    \path (X0) ++(-90:2cm) node[state] (SnX0) {$\rhd S{\bar x}_i$};

    \path (X0) ++(2cm,-12mm) node[state] (nx00) {${\bar x}_i^0$} ++(0:2cm) node[state] (nx01) {${\bar x}_i^1$} ++(0:2cm) node[state] (nx02) {${\bar x}_i^2$} ++(0:3cm) node[state] (nx0l) {${\bar x}_i^{\ell}$};
    
    \path (nx0l) ++(2cm,12mm) node[state] (endX0) {$x_i \rhd$};

    \path (endX0) ++(90:2cm) node[state] (endSX0) {$Sx_i\rhd $};    
    \path (endX0) ++(-90:2cm) node[state] (endSnX0) {$S{\bar x}_i\rhd $};

    \path (endSX0) ++(90:3cm) node[out state] (end0) {$end_{x_i}$};   
    \path[p2,out transition] (endSX0) to[bend right] node[right,midway] {$5$} (end0);
    \path[p3,out transition] (end0) to[bend right] node[left,midway] {$7$} (endSX0);
    
    \path (SX0) ++(90:3cm) node[out state] (start0) {$start_{x_i}$};   
    \path[p3,out transition] (SX0) to[bend right] node[right,midway] {$7$} (start0);
    \path[p2,out transition] (start0) to[bend right] node[left,midway] {$5$} (SX0);

    \path (endSnX0) ++(-90:3cm) node[out state] (endn0) {$end_{\overline{x_i}}$};   
    \path[p2,out transition] (endSnX0) to[bend left] node[right,midway] {$5$} (endn0);
    \path[p3,out transition] (endn0) to[bend left] node[left,midway] {$7$} (endSnX0);
    
    \path (SnX0) ++(-90:3cm) node[out state] (startn0) {$start_{\overline{x_i}}$};   
    \path[p3,out transition] (SnX0) to[bend left] node[right,midway] {$7$} (startn0);
    \path[p2,out transition] (startn0) to[bend left] node[left,midway] {$5$} (SnX0);

    \path(endX0) ++(0:3cm) node[out state] (X1) {$\rhd x_{i+1}$};
    \path(endSX0) ++(0:3cm) node[out state] (SX1) {$\rhd Sx_{i+1}$};
    \path(endSnX0) ++(0:3cm) node[out state] (SnX1) {$\rhd S{\bar x}_{i+1}$};
        
    \path(X0) ++(180:3cm) node[out state] (endXm1) {$x_{i-1}\rhd$};
    \path(SX0) ++(180:3cm) node[out state] (endSXm1) {$Sx_{i-1}\rhd$};
    \path(SnX0) ++(180:3cm) node[out state] (endSnXm1) {$S{\bar x}_{i-1}\rhd$};

    \path[p1,transition] (X0) to node[above,midway] {$1$} (x00);
    \path[p1,transition] (x00) to node[above,midway] {$1$} (x01);
    \path[p3,transition] (x00) to[bend angle=60,bend left] node[above,midway] {$6$} (x01);
    \path[p2,transition] (x01) to[bend angle=30,bend left] node[below,midway] {$4$} (x00);
    \path[p1,transition] (x01) to node[above,midway] {$1$} (x02);
    \path[p3,transition] (x01) to[bend angle=60,bend left] node[above,midway] {$6$} (x02);
    \path[p2,transition] (x02) to[bend angle=30,bend left] node[below,midway] {$4$} (x01);
    \path (x02) to node[midway]{$\cdots$} (x0l);
    \path[p1,transition] (x0l) to node[above,midway] {$1$}  (endX0);
    \path[p2,transition] (SX0) to node[above,midway] {$5$}  (x00);
    \path[p2,transition] (x0l) to node[above,midway] {$5$}  (endSX0);

    \path[p1,transition] (X0) to node[below,midway] {$1$} (nx00);
    \path[p1,transition] (nx00) to node[below,midway] {$1$} (nx01);
    \path[p3,transition] (nx00) to[bend angle=60,bend right] node[below,midway] {$6$} (nx01);
    \path[p2,transition] (nx01) to[bend angle=30,bend right] node[above,midway] {$4$} (nx00);
    \path[p3,transition] (nx01) to[bend angle=60,bend right] node[below,midway] {$6$} (nx02);
    \path[p2,transition] (nx02) to[bend angle=30,bend right] node[above,midway] {$4$} (nx01);
    \path[p1,transition] (nx01) to node[below,midway] {$1$} (nx02);
    \path (nx02) to node[midway]{$\cdots$} (nx0l);
    \path[p1,transition] (nx0l) to node[below,midway] {$1$}  (endX0);
    \path[p2,transition] (SnX0) to node[below,midway] {$5$}  (nx00);
    \path[p2,transition] (nx0l) to node[below,midway] {$5$}  (endSnX0);

    \path[draw] ($(SX0.north west)+(-5mm,5mm)$) -- ($(endSX0.north east)+(5mm,5mm)$) -- ($(endSnX0.south east)+(5mm,-5mm)$) -- ($(SnX0.south west)+(-5mm,-5mm)$) -- cycle;

    \foreach\i in {01,02,0l}{
        \path(x\i) ++(-3mm,30mm) coordinate (tmp1){};
        \path[p1, out transition,dotted] (x\i) to node[near end,left] {$3$} (tmp1);
    }
    \foreach\i in {01,02,00}{
        \path(x\i) ++(3mm,30mm) coordinate (tmp2){};
        \path[p1, out transition,dotted] (tmp2) to node[near start,right] {$3$} (x\i);
    }
    \path ($(tmp1)!.5!(tmp2)$) node[anchor=south] {Possible edges to/from some gadget $C_j$};

    \foreach\i in {01,02,0l}{
        \path(nx\i) ++(-3mm,-30mm) coordinate (tmp1){};
        \path[p1, out transition,dotted] (nx\i) to node[near end,left] {$3$} (tmp1);
    }
    \foreach\i in {01,02,00}{
        \path(nx\i) ++(3mm,-30mm) coordinate (tmp2){};
        \path[p1, out transition,dotted] (tmp2) to node[near start,right] {$3$} (nx\i);
    }
    \path ($(tmp1)!.5!(tmp2)$) node[anchor=north] {Possible edges to/from some gadget $C_j$};

    \path[p1, out transition] (endX0) to node[above,midway] {2} (X1);  
    \path[p3, out transition] (endSX0) to node[above,midway] {8} (SX1);      \path[p3, out transition] (endSnX0) to node[above,midway] {8} (SnX1);  
    
    \path[p1, out transition] (endXm1) to node[above,midway] {2} (X0);
    \path[p3, out transition] (endSXm1) to node[above,midway] {8} (SX0);
    \path[p3, out transition] (endSnXm1) to node[above,midway] {8} (SnX0);
\end{tikzpicture}}
    \caption{Gadget for variable $x_i$, $i\in\{1,\ldots,k\}$}
    \label{fig:gadget-var}
\end{figure}
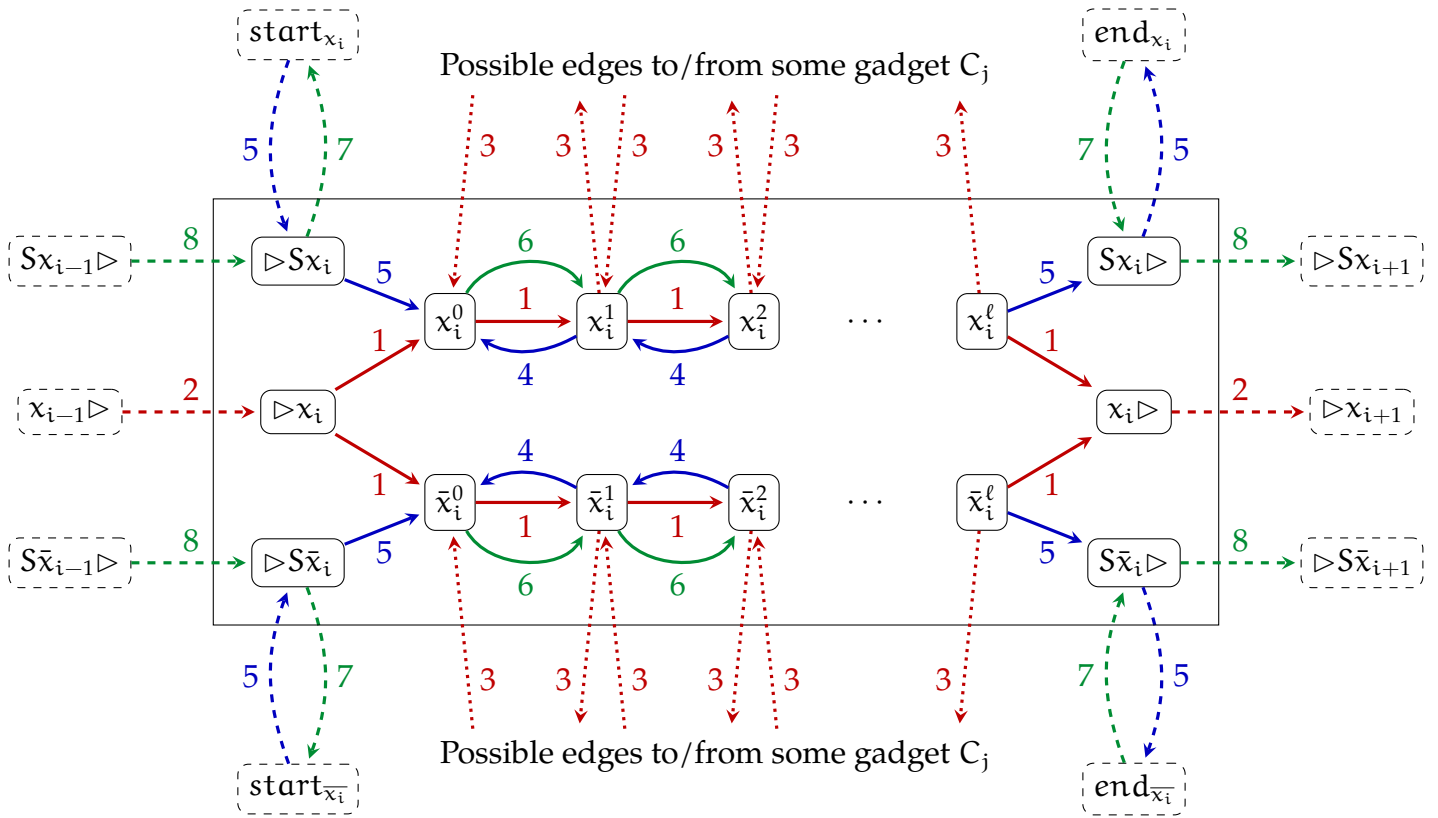

For each clause~$C_j$, there is one \defemph{clause gadget} defined in Figure~\ref{fig:gadget-cl}.
It contains two vertices, is linked to the prior and following clause gadget.
As said before, it is connected by two edges to a variable gadget~$x_i$ if~$x_i$ appears
in~$C_j$. It is connected to the positive (resp.~negative) side if it appears positively (resp.~negatively) in~$C_j$.
    
\begin{figure}[hpt]\centering
    \begin{tikzpicture}

    \path node[state] (-Cj) {$\rhd C_j$}
        ++(80mm,0cm) node[state] (Cj-) {$C_j\rhd$};

    \path[gadget] let \p1=(current bounding box.north east), \p2=(current bounding box.south west) in
        (\x1+5mm,\y1+5mm) rectangle (\x2-5mm,\y2-5mm);

    \path (-Cj) ++(-30mm,0) node[out state,] (Cjm1-) {$C_{j-1}\rhd $};
    \path[p1, out transition] (Cjm1-) to node[above] {$2$} (-Cj);

    \path (Cj-) ++(30mm,0) node[out state] (-Cjp1) {$\rhd C_{j+1}$};
    \path[p1, out transition] (Cj-) to node[above] {$2$} (-Cjp1);

    \path (-Cj) 
        ++(25mm,20mm)
        node[out state] (-aj1) {$\beta^{j-1}$}
        ++(0,15mm)
        node[out state] (-aj2) {$\gamma^{j-1}$}
        ++(0,15mm)
        node[out state] (-aj3) {$\delta^{j-1}$};

    \path[out transition, p1, rounded corners] (-Cj.120) |- node[above,near end] {$3$} (-aj3);
    \path[out transition, p1, rounded corners] (-Cj.90) |- node[above,near end] {$3$} (-aj2);
    \path[out transition, p1, rounded corners] (-Cj.60) |- node[above,near end] {$3$} (-aj1);

    \path (Cj-) 
        ++(-25mm,20mm)
        node[out state] (aj1-) {$\beta^{j}$}
        ++(0,15mm)
        node[out state] (aj2-) {$\gamma^{j}$}
        ++(0,15mm)
        node[out state] (aj3-) {$\delta^{j}$};

    \path[out transition, p1, rounded corners] (aj3-) -| node[near start,above] {$3$} (Cj-.60);
    \path[out transition, p1, rounded corners] (aj2-) -| node[near start,above] {$3$} (Cj-.90);
    \path[out transition, p1, rounded corners] (aj1-) -| node[near start,above] {$3$} (Cj-.120) ;

\end{tikzpicture}
    \caption{Gadget for the clause $C_j$, $j\in\{1,\ldots,\ell\}$, decomposed as~$C_j=\beta \vee \gamma \vee \delta$ with $\beta,\gamma,\delta\in\{x_1,\ldots,x_n\}\cup\{\bar{x}_1,\ldots,\bar{x}_n\}$.
    For instance, if~$\beta=x_2$ (resp.~${\bar x}_5$), then $\beta^j$ refers to the vertex $x_2^j$ (resp.~${\bar x}_5^j$) in gadget $x_2$ (resp. in gadget~$x_5$).}
    \label{fig:gadget-cl}
\end{figure}
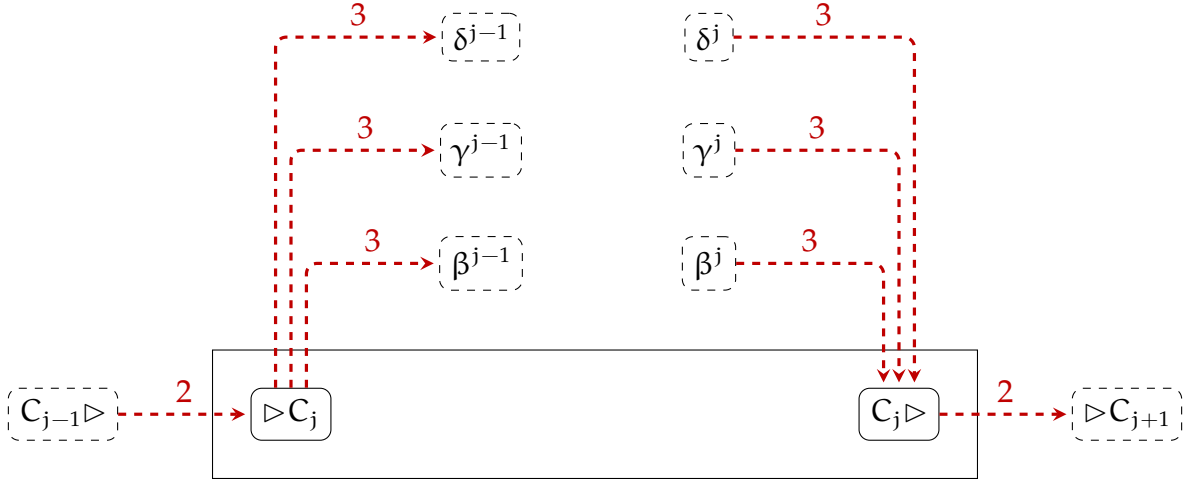

Figure~\ref{fig:gadget-end} defines the \defemph{end gadget}
Aside from the vertex $Target$, which will be the target of all matched walks, it contains 2 vertex for each variable~$x_i$.
They are connected to the end of the positive and negative side of gadget $x_i$, respectively.
        
\begin{figure}[h]\centering
    \makebox[\linewidth][c]{\begin{tikzpicture}
    \path node[out state] (Cl) {$C_{\ell+1}$}
        ++(0:3cm) node[state] (end1) {$end_{x_1}$}
        ++(0:24mm) node[state] (endn1) {$end_{\overline{x_1}}$}
        % ++(0:24mm) node[state] (end2) {$end_2$}
        % ++(0:24mm) node[state] (endn2) {$\overline{end}_2$}
        ++(0:3cm) node[state] (endk) {$end_{x_k}$}
        ++(0:24mm) node[state] (endnk) {$end_{\overline{x_k}}$}
        ++(0:34mm) node[state] (end) {$Target$}
        ;
        
    \path[p1, out transition] (Cl) to node[above] {$2$} (end1);
    \path[p1, transition] (end1) to node[above] {$2$} (endn1);
    % \path[p1, transition] (endn1) to node[above] {$2$} (end2);
    % \path[p1, transition] (end2) to node[above] {$2$} (endn2);
    \path (endn1) to node{$\cdots$} (endk);
    \path[p1, transition] (endk) to node[above] {$2$} (endnk);
    \path[p1, transition] (endnk) to node[above] {$0$} (end);

    \path[draw] ($(end1.north west)+(-5mm,5mm)$) -- ($(end.north east)+(5mm,5mm)$) -- ($(end.south east)+(5mm,-5mm)$) -- ($(end1.south west)+(-5mm,-5mm)$) -- cycle;
    
    \path (end1) ++(90:3cm) node[out state] (SX1) {$Sx_1\rhd $};   
    \path[p3,out transition] (end1) to[bend right] node[right,midway] {$7$} (SX1);
    \path[p2,out transition] (SX1) to[bend right] node[left,midway] {$5$} (end1);

    \path (endn1) ++(90:3cm) node[out state] (SnX1) {$S{\bar x}_1\rhd $};   
    \path[p3,out transition] (endn1) to[bend right] node[right,midway] {$7$} (SnX1);
    \path[p2,out transition] (SnX1) to[bend right] node[left,midway] {$5$} (endn1);

    \path (endk) ++(90:3cm) node[out state] (SXk) {$Sx_k\rhd $};   
    \path[p3,out transition] (endk) to[bend right] node[right,midway] {$7$} (SXk);
    \path[p2,out transition] (SXk) to[bend right] node[left,midway] {$5$} (endk);

    \path (endnk) ++(90:3cm) node[out state] (SnXk) {$S{\bar x}_k\rhd $};   
    \path[p3,out transition] (endnk) to[bend right] node[right,midway] {$7$} (SnXk);
    \path[p2,out transition] (SnXk) to[bend right] node[left,midway] {$5$} (endnk);

    \path (end) ++(-10mm,3cm) node[out state] (Sx0) {$S{x_0}\rhd$};   
    \path[p3,out transition] (end) to[bend left] node[left,midway] {$9$} (Sx0);    
    \path (end) ++(10mm,3cm) node[out state] (Sxkp1) {$\rhd S\bar{x}_{k+1}$};   
    \path[p3,out transition]  (Sxkp1) to[bend left]  node[right,midway] {$X$}(end);
\end{tikzpicture}}
    \caption{End gadget}
    \label{fig:gadget-end}
\end{figure}
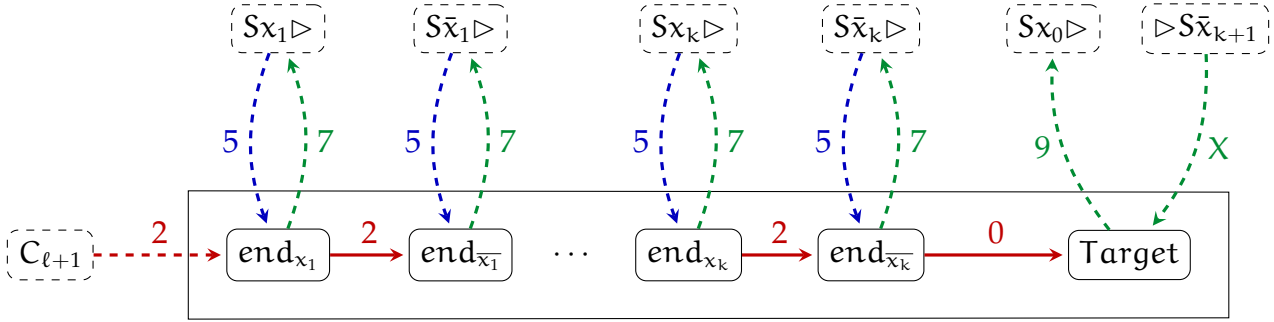

Finally, Figure~\ref{fig:gadget-glu} defines a few missing vertices that connect some gadgets to each other.
First, it links with red edges the start gadget to the first variable gadget,
then the last variable gadget to the first clause gadget,
and the last clause gadget to the end gadget.
Second, it helps create with green edges a cycle from $Target$ to $Target$ that traverse the positive sides and the negative sides of all variable gadget.

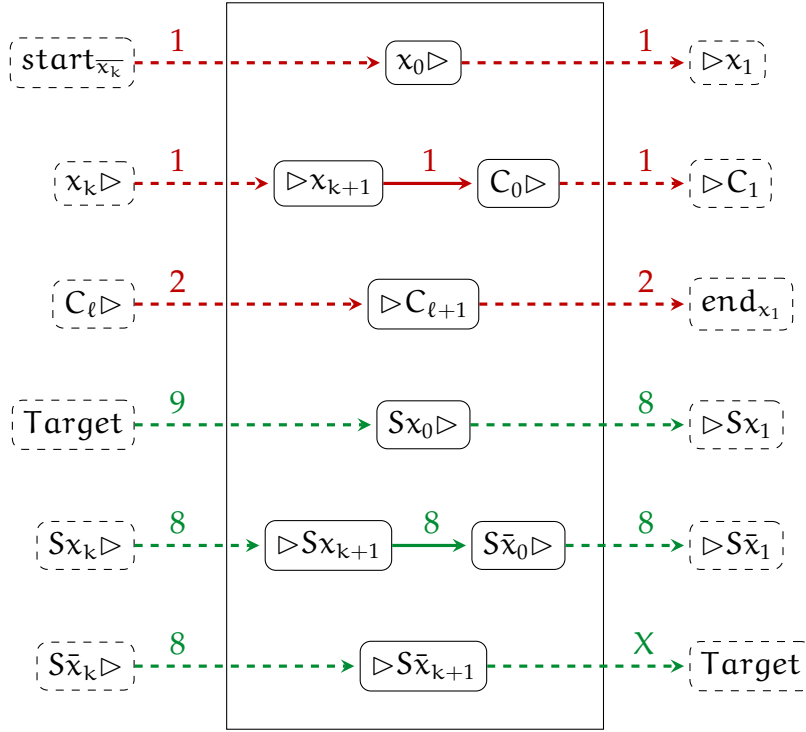
\begin{figure}[h]\centering
    \begin{tikzpicture}
    \path 
        ++(0:25mm) node[state] (xstart) {$x_0\rhd$}
        ++(-12.5mm,-16mm) node[state] (xend) {$\rhd x_{k+1}$}
        ++(25mm,0) node[state] (Cstart) {$C_0\rhd $}
        ++(-12.5mm,-16mm) node[state] (Cend) {$\rhd C_{\ell+1}$}
        ++(0,-16mm) node[state] (Sxstart) {$Sx_0\rhd$}
        ++(-12.5mm,-16mm) node[state] (Sxend) {$\rhd Sx_{k+1}$}
        ++(0:25mm) node[state] (Snxstart) {$S{\bar x}_0\rhd$}
        ++(-12.5mm,-16mm) node[state] (Snxend) {$\rhd S{\bar x}_{k+1}$}
    ;
    \path[gadget] let \p1=(current bounding box.north east), \p2=(current bounding box.south west) in
        (\x1+5mm,\y1+5mm) rectangle (\x2-5mm,\y2-5mm);

    \path (xstart) ++(-38mm,0) node[out state, anchor=east] (nstartk) {$start_{\overline{x_k}}$};
    \path[p1,out transition] (nstartk) to node[above right, at start, xshift=3mm] {$1$} (xstart);

    \path (xstart) ++(35mm,0) node[out state, anchor=west] (x1) {$\rhd x_1$};
    \path[p1,out transition] (xstart) to node[above left, at end, xshift=-3mm] {$1$} (x1);

    \path (xend) ++(-25.5mm,0) node[out state, anchor=east] (xk) {$x_{k}\rhd$};
    \path[p1,out transition] (xk) to node[above right, at start, xshift=3mm] {$1$} (xend);
    
    \path[p1,transition] (xend) to node[above,midway] {$1$} (Cstart);
    
    \path (Cstart) ++(22.5mm,0) node[out state, anchor=west] (C1) {$\rhd C_1$};
    \path[p1,out transition] (Cstart) to node[above left, at end, xshift=-3mm] {$1$} (C1);

    \path (Cend) ++(-38mm,0) node[out state, anchor=east] (Cl) {$C_\ell\rhd$};
    \path[p1,out transition] (Cl) to node[above right, at start, xshift=3mm] {$2$} (Cend);

    \path (Cend) ++(35mm,0) node[out state, anchor=west] (end1) {$end_{x_1}$};
    \path[p1,out transition] (Cend) to node[above left, at end, xshift=-3mm] {$2$} (end1);

    \path (Sxstart) ++(-38mm,0) node[out state, anchor=east] (End) {$Target$};
    \path[p3,out transition] (End) to node[above right, at start, xshift=3mm] {$9$} (Sxstart);
    
    \path (Sxstart) ++(35mm,0) node[out state, anchor=west] (SX1) {$\rhd Sx_1$};
    \path[p3,out transition] (Sxstart) to node[above left, at end, xshift=-3mm] {$8$} (SX1);

    \path (Sxend) ++(-25.5mm,0) node[out state, anchor=east] (Sxk) {$Sx_k\rhd$};
    \path[p3,out transition] (Sxk) to node[above right, at start, xshift=3mm] {$8$} (Sxend);
    
    \path[p3,transition] (Sxend) to node[above,midway] {$8$} (Snxstart);   
    
    \path (Snxstart) ++(22.5mm,0) node[out state, anchor=west] (SnX1) {$\rhd S{\bar x}_1$};
    \path[p3,out transition] (Snxstart) to node[above left, at end, xshift=-3mm] {$8$} (SnX1);

    \path (Snxend) ++(-38mm,0) node[out state, anchor=east] (Snxk) {$S{\bar x}_k\rhd$};
    \path[p3,out transition] (Snxk) to node[above right, at start, xshift=3mm] {$8$} (Snxend);
    
    \path (Snxend) ++(35mm,0) node[out state, anchor=west] (End2) {$Target$};
    \path[p3,out transition] (Snxend) to node[above left, at end, xshift=-3mm] {$X$} (End2);
    
\end{tikzpicture}
    \caption{Glu gadget}
    \label{fig:gadget-glu}
\end{figure}

\subsection{Proof of Theorem~\ref{th:enum-hardness}}
\label{sec:real-proof}

First, let us consider the walks matching~$R_1$.
Recall that  $R_1=\Ri$.
A simple verification yields the following lemma.
\begin{lemma}\label{l:match-r1}
    Any walk~$w$ matching~$R_1$ meets the following properties.
    \begin{subthm}
        \item All edges in~$w$ are red, that is labelled over $\{0,1,2,3\}$.
        \item The source of $w$ is $Source$.
        \item \label{l:match-r1:side} The walk~$w$ traverse once each gadget $x_i$ and 
        exactly one of the following walks is a factor of~$w$.
        \begin{align*}
            &\rhd x_i \rightarrow x_i^0\rightarrow\cdots\rightarrow x_i^\ell\rightarrow x_i\rhd \\
            & \rhd x_i \rightarrow {\bar x}_i^0\rightarrow\cdots\rightarrow {\bar x}_i^\ell\rightarrow x_i\rhd 
        \end{align*}
            If the former is a factor of~$w$, we say that~$w$ \defemph{traverses the positive side of gadget $x_i$},
            and or that it \defemph{traverses the negative side of gadget $x_i$} otherwise.
        \item The walk~$w$ traverse once each gadget $C_j$.
        \item The target of~$w$ is $Target$.
        \item \label{l:match-r1:length} The length \footnote{The precise length of~$w$ is unimportant.  In the following we use only that all walks matching $R_1$ have the same length, and that this length is polynomial in~$I$.} of~$w$ is equal to~$k\ell+7k+4\ell+3$.
    \end{subthm}
\end{lemma}

\begin{definition}\label{d:valuation}
    With any walk~$w$ matching~$R_1$ is associated a \emph{valuation} $f_w$,
    the one that maps each variable~$x_i$ to $\top$ if it traverses the negative side of the gadget~$x_i$, and to~$\bot$ if it traverses its positive side.

    Be aware that sign of~$x_i$ in the valuation~$f_w$ is the opposite of the sign of the side traversed in gadget~$x_i$.
\end{definition}

\begin{proposition}\label{p:3sat<=>trail}
    There exists a trail in~$\Gc_I$ matching $R_1$ if and only if the 3-\problemfont{Sat} instance $I$ is satisfiable.
\end{proposition}
\begin{proof}
    $\Rightarrow$ Let~$w$ be a trail matching $R_1$.  
    Let us show that the valuation~$f_w$ associated with~$w$ makes~$I$ true.
    Let~$C_j=\beta \vee \gamma\vee \delta$ be a clause in~$I$. Then the walk~$w$ must contain the factor $\rhd C_j\xrightarrow{3}\alpha^{j-1}\xrightarrow{1}\alpha^{j}\xrightarrow{3}C_j\rhd$ for some~$\alpha\in\{\beta,\gamma,\delta\}$.
    Let us assume that $\alpha$ is a positive literal~$x_i$, which means that $w$ contains 
    $C_j\xrightarrow{3} x_i^{j-1} \xrightarrow{1} x_i^{j} \xrightarrow{3} C_{j+1}$ as a factor.
    Since~$w$ is a trail, it means that $w$ necessarily traverses the negative side of gadget $x_i$, hence that $f_w(x_i)=\top$. Hence $f_w$ indeeds makes clause $C_j$ true. The case where~$\alpha$ is a negative literal is symmetric.
    We just showed that $f_w$ makes true every clause of~$I$, hence~$I$ is satisfiable.
    
    The $\Leftarrow$ direction is similar.  The valuation making~$I$ true gives the side used in each variable gadget and since it must make~$I$ true, it also provides implicitly which variable(s) makes each clause true.
\end{proof}

Recall that  $R_2=\Rii$.  
Let us now precisely characterises the walks matching~$R_2$.  

\begin{lemma}\label{l:match-r2}
    The set of walks in~$\Gc_I$ matching $R_2$ is exactly
    \begin{equation*}
        \{w_{a,j} \mid \alpha\in\{x_1,\ldots,x_k\}\cup\{{\bar x}_1,\ldots,{\bar x}_k\},~ j\in\{1,\ldots,\ell\} \}
    \end{equation*}
    where~$w_{\alpha,j}$ is the following walk.
    \begin{multline*}
        Source \xrightarrow{0} Start_{x_1}\xrightarrow{2} \cdots \xrightarrow{2} Start_\alpha \xrightarrow{5} \\
            \rhd S\alpha \xrightarrow{5} \alpha^{0} \xrightarrow{1} \cdots \xrightarrow{1} \alpha^{j-1} \xrightarrow{1} \alpha^{j} \xrightarrow{4} \alpha^{j-1} \xrightarrow{1} \alpha^{j} \xrightarrow{5} S\alpha\rhd \xrightarrow{5} \\
            End_\alpha \xrightarrow{2} \cdots \xrightarrow{2} End_{\bar{x}_1} \xrightarrow{0} Target
    \end{multline*}
\end{lemma}
\begin{proof}[Sketch of proof]
    The fact that all such walks are matching~$R_2$ requires a simple verification.
    Conversely, each walk matching~$R_2$ is necessarily of this form because of the mandatory $4$ and $55$ in $R_2$. Indeed, $55$ forces the walk to choose a variable gadget and the positive or negative side of that gadget.
    Then the mandatory letter~$4$ forces the matching walk to repeat an edge within that gadget, the repeated edge fixes~$j$.
\end{proof}

Note that the walks~$w_{\alpha,j}$ is almost a trail: it repeats the edge $\alpha^{j-1} \xrightarrow{1} \alpha^{j}$ twice but every other edge appears once.
Moreover, it follows from Lemma~\thelemma{} that they satisfy the following properties.

\begin{corollary}\label{c:match-r2}
    There are $2\times \ell \times k$ walks in~$\Gc_I$ matching~$R_2$.
    They all belong to $\minmultsem(\Gc_I,R_2)$.
    They all have a length smaller than $4k+\ell+6$.
\end{corollary}

\begin{proposition}\label{p:red-edge-inclusion}
    For every walk~$w_1$ matching $R_1$, the following are equivalent.
    \begin{subthm}
        \item $w_1$ is not a trail;
        \item there exists a walk $w_2$ matching $R_2$ such that
        $\rededges(w_2)\subseteq \rededges(w_1)$,
    where $\rededges$ maps a walk~$w$ to the multiset of the red edges in~$w$, that is the edges labeled by 0, 1, 2 or 3.
    \end{subthm}
\end{proposition}     
\begin{proof}
    $\enumstyle{a}\implies\enumstyle{b}$
    If~$w_1$ is not a trail, $w_1$ uses twice some red edge $\alpha^{j-1}\xrightarrow{1}\alpha^{j}$ with $\alpha\in\{x_i,\bar{x}_i\}$ for some variable~$x_i$.
    Then the walk $w_2=w_{\alpha,j}$ (see Lemma~\ref{l:match-r2}) uses that red edge twice. 
    Note that the walk~$w_1$ contains each red edges from the start and end gadgets.
    Since~$w_{\alpha,j}$ traverses only three gadgets, it remains to show that all red edges used by $w_{\alpha,j}$ in the gadget $x_i$ are also in~$w_1$.
    The only way for $w$ to use~$\alpha^{j-1}\xrightarrow{1}\alpha^{j}$ twice is that
    $w_1$ contains the factor $\alpha^0\xrightarrow{1}\cdots\xrightarrow{1} \alpha^j$ (while traversing gadget~$x_i$, see item~\ref{l:match-r1:side}).
    
    $\enumstyle{b}\implies\enumstyle{a}$
    Every match to~$R_2$ repeats a red edge, the hypothesis $\rededges(w_2)\subseteq \rededges(w_1)$ implies that~$w_1$ repeats that same red edge, hence is not a trail.
\end{proof}

We now consider the expression $R_3=\Riii$. A simple verification yields the following characterization of the walks (in fact the unique walk) that matches $R_3$.

\begin{lemma}\label{l:match-r3}
    There is exactly one walk~$w_3$ matching~$R_3$, and it satisfies the following properties
    \begin{subthm}
        \item \label{l:match-r3:endpoints} $w_3$ starts and ends in $Target$.
        % \item \label{l:match-r3:trail} $w_3$ is a trail.
        \item \label{l:match-r3:one-copy} $w_3$ contains one or two copies of each blue or green edge, that is of each edge labeled over \{$4$, $5$, $6$, $7$, $8$, $9$, $X$\}.
    \end{subthm}
\end{lemma}

In particular, note that for every match~$w_1$ to~$R_1$, the walk $w_1\cdot w_3$ matches~$R_1R_3$, where~$w_3$ is the match to~$R_3$.

\begin{proposition}\label{p:equiv}
    Let~$w_1$ be a match to $R_1$
    and $w_3$ be the match to $R_3$. The following are equivalent.
    \begin{subthm}
        \item $w_1$ is not a trail;
        \item There exists a walk $w_2$ matching $R_2$ such that
        $\edges(w_2)\subsetneq \edges(w_1\cdot w_3)$;
        \item $w_1\cdot w_3\notin \minmultsem(D,R_1\cdot R_3+R_2)$;
    \end{subthm}
\end{proposition}
\begin{proof} 
It follows from Lemma~\ref{l:match-r3:one-copy} and Proposition~\ref{p:red-edge-inclusion} that \enumstyle{a} and \enumstyle{b} are equivalent.
Definition of~$\minmultsem$ implies that $\enumstyle{b} \implies  \enumstyle{c}$. 

Finally, let us show $\enumstyle{c}\implies \enumstyle{b}$.
By definition of $\minmultsem$, $w_1\cdot w_3 \notin \minmultsem(D,R_1\cdot R_3+R_2)$
means that there exists a walk $w_2$ matching $\R$ such that $\edges(w_2)\subsetneq \edges(w_1\cdot w_3)$.
We assume, for the sake of contradiction that~$w_2$ matches $R_1\cdot R_3$.
Recall that all matches to~$R_1$ (resp.~$R_3$) have the same length from Lemma~\ref{l:match-r1:length} (resp.~from Lemma~\ref{l:match-r3}).
Hence, the length of~$w_2$ is equal to the length of~$w_1w_3$, a contradiction to the condition $\edges(w_2)\subsetneq \edges(w_1\cdot w_3)$.
Since~$w_2$ matches $\R$ but not $R_1R_3$, then it must match $R_2$, hence \enumstyle{b} holds. 
\end{proof}

We may now finalise the proof of Theorem~\ref{th:enum-hardness}.
From Corollary \ref{c:match-r2}, there are $2 \ell  k$ walks in~$\minmultsem(\Gc_I,R_2)$.
They must all belong to~$\minmultsem(\Gc_I,R)$ since it is impossible that $\edges(w)\subsetneq\edges(w')$ with $w$ matching $R_1R_3$ and $w'$ matching~$R_2$ ($w$ contains green edges, $w'$ does not).
Now, let us assume there is a polynomial-delay algorithm~$A$ to enumerate~$\minmultsem(\Gc_I,R)$.
Note that the following are equivalent.
\begin{itemize}
    \item[] The algorithm~$A$ outputs at least $(2 \ell  k+1)$-th walks.
    \item[] $\iff$ There exists a walk~$w\in\minmultsem(\Gc_I,R)$ that matches $R_1\cdot R_3$.
    \item[] $\iff$ There is a trail~$w$ matching $R_1$. (Proposition~\ref{p:equiv})
    \item[] $\iff$ $I$ is satisfiable. (Proposition~\ref{p:3sat<=>trail})
\end{itemize}
On the other hand, it follows from Lemma~\ref{l:match-r1:length}, Corollary~\ref{c:match-r2} and Lemma~\ref{l:match-r3} that all matches to~$R=\R$ have a length that is polynomial in~$k$ and~$\ell$.
Hence, $A$ takes a polynomial time to output $2 \ell  k+1$ walks or to stop
after outputting~$2 \ell  k$.
This provides a polynomial time algorithm for 3-SAT: build~$\Gc_I$ and run $A$ until it outputs $2 \ell  k+1$ walks or stops at~$2 \ell  k$;
in the first case accept $I$, in the second reject~$I$.

%%%%%%%%%%%%%%%%%%%%%%%%%%%%%%%%%%%%%%%%%%%%%%%%%%%%%%%%%%%%%%%%%%%%%%%%%%%%%%%%
%%%%%%%%%%%%%%%%%%%%%%%%%%%%%%%%%%%%%%%%%%%%%%%%%%%%%%%%%%%%%%%%%%%%%%%%%%%%%%%%
%%%%%%%%%%%%%%%%%%%%%%%%%%%%%%%%%%%%%%%%%%%%%%%%%%%%%%%%%%%%%%%%%%%%%%%%%%%%%%%%
\section{Side results}
\label{s:side-result}

\subsection{Walk membership is co-NP hard}

\begin{problem}{Walk Membership under \mm}
    \item[Input:] A regular expression~$R$, a graph graph~$\Gc_I$ and a walk~$w$.
    \item[Question:] Does $w$ belong to $\minmultsem(R,D)$?
\end{problem}

\begin{theorem}\label{th:memb-conp-compl-mm}
    \problemfont{Walk membership} under $\minmultsem$ semantics is \conp-complete.
    It is already true for a fixed regular expression~$R$.
\end{theorem}
\begin{proof}[Sketch of proof]
    \problemfont{Walk Membership}$(\Gc,R,w)$ under $\minmultsem$ belongs to \conp since a certificate for rejection is any walk~$w'$ in $\walks(\Gc,R)$ such that $\mm(w')\subsetneq \mm(w)$.
    
    The proof of \conp-hardness is similar, but simpler than the one of Theorem~\ref{th:enum-hardness}.
    We use same red part of the graph and the same expression~$R_1=\Ri$.
    Hence, Proposition~\ref{p:3sat<=>trail} still holds: each walk matching $R_1$ is a trail if and only if the 3-SAT instance is satisfiable .
    Then we add green edges to the graph. This time, we need to devise an expression~$R_2$ that matches a walk~$w_2$ such that every red edge occurs exactly once in~$w_2$ (green edges may occur any number of times).
    This may for instance be done as follows.
    \begin{itemize}
        \item We add a single green letter~$G$ to the alphabet.
        \item We add an edge $x_i \rhd \xrightarrow{G} \rhd x_i$ for every~$j\in\{1,\ldots,k\}$.
        \item We add an edge $C_j \rhd \xrightarrow{G} \rhd C_j$ for every~$j\in\{1,,\ldots,\ell\}$.
        \item An edge $\alpha^{j-1} \rhd \xrightarrow{G} \alpha^{j}$ for every~$j\in\{1,\ell\}$ and every literal~$\alpha$ in clause $C_j$.
        \item Using the expression~$R_2=0\cdot (11^*G11^*+2+(3G)^53)^*\cdot 0$.
    \end{itemize}
    The pattern $11^*G11^*$ inside the star of~$R_2$ allows some walks matching~$R_2$ to have as factors both the positive and the negative 1-walk of each variable gadget (see item \ref{l:match-r1:side}).
    Similarly, the pattern $(3G)^53$ allows some walks matching~$R_2$ to pass through every 3-edge in each clause gadget.
    In particular, there is a match~$w_2$ to~$R_2$ that contains every red edge.
    Note also that all walks matching~$R_2$ have the same length, hence it is impossible that a walks~$w_2'$ matching $R_2$ satisfies $\edges(w_2)\subsetneq\edges(w_2')$.
    Then~$w$ belongs to $\minmultsem(\Gc,R_1+R_2)$ if and only if $I$ is not satisfiable.
\end{proof}

\subsection{Minimal-Set Semantics}
\newcommand{\smsle}{\sqsubset}
\newcommand{\smsleq}{\sqsubseteq}

As mentioned in the introduction, the problem studied in this note were introduced in \cite{MarsaultMeyer2024}.
Other variants were also mentioned, and we show how the proof can easily be adapted for some of them.

\begin{definition}
    \begin{subthm}
        \item \label{def:minimal-multiset}
        For every walk~$w$, we let~$\edgeset(w)$ denote the \textbf{set} of the edges appearing in~$w$.
        \item 
        We define the strict partial order relation~$\smsle$ on walks as follows:~$w\smsle w'$ if and only if: either $\edgeset(w)\subsetneq \edgeset(w')$;  or $\edgeset(w) = \edgeset(w')$ and $\len(w)<\len(w')$.
        We let~$\smsleq$ denote the reflexive extension of~$\smsle$.
        \item We let~$\sms$ denote the function that maps a graph~$\Gc$, a regular expression~$R$,
        and two vertices vertex~$s,t$ to:
        \begin{equation}
            \sms(\Gc,R,s,t) = \min_{\smsle} \matches(\Gc,R,s,t)\quad.
        \end{equation}
    \end{subthm}
\end{definition}

In Item \enumstyle{b} above, we use length as a tie-breaker when two walks have equal edge sets.  
Indeed, recall that without such a tie breaker, \sms{} would not always return a finite set of walks.
It may be verified that with the current definition, $\sms(\Gc,R,s,t)$ is indeed always finite.

\begin{problem}{Walk Enum Under \sms}
  \item[Input:] A regular expression~$R$, a graph~$\Gc$ and two vertices~$s,t$.
  \item[Output:] All walks in $\sms(R,\Gc,s,t)$.
\end{problem}

\begin{theorem}\label{th:enum-hardness-sms}
        Unless $\ptime=\np$, there exists no polynomial delay algorithm to solve \problemfont{Walk Enum Under \sms}. It is already true for a fixed regular expression~$R$ with star-height~$1$.
\end{theorem}

The proof of Theorem \ref{th:enum-hardness-sms} is very similar to the one of theorem~\ref{th:enum-hardness} and we only sketch the difference below.
First, we need to modify the graph and expression as follows.
\begin{itemize}
    \item Add a new blue letter~$4'$.
    \item Replace each edge $x_i^{j} \xrightarrow{4} x_i^{j-1}$ by two edges $x_i^{j} \xrightarrow{4} \bar{x}_i^{j-1}$ and ${\bar x}_i^{j} \xrightarrow{4'} x_i^{j}$.
    \item Symmetrically replace each edge ${\bar x}_i^{j} \xrightarrow{4} {\bar x}_i^{j-1}$ by two edges ${\bar x}_i^{j} \xrightarrow{4} x_i^{j-1}$ and $x_i^{j} \xrightarrow{4'} {\bar x}_i^{j}$.
    \item $R_2$ becomes $0\cdot 2^* 55\cdot 1^* \cdot 414' \cdot 1^* \cdot 55 \cdot 2^*$
    \item $R_3$ becomes $9\cdot (8+755+6464'+557)^*\cdot X$
\end{itemize}
Second, a proof much similar to the one of Proposition~\ref{p:3sat<=>trail} shows that $I$ is satisfiable if and only if there is a walk matching~$R_1$ that does not contain both the edges $x_i^{j} \xrightarrow{1} x_i^{j-1}$ and $x_i^{j} \xrightarrow{1} x_i^{j-1}$.
The matches to $R_2$ are exactly the walks using one side (positive or negative) from a variable gadget, and one edge from the other side of this gadget. The remainder of the proof is similar to the one of Theorem~\ref{th:enum-hardness}.

\begin{remark}
    The proof of Theorem~\ref{th:enum-hardness-sms} does not actually use the fact that $\smsle$ 
    uses length a a tiebreak.
    Hence Theorem~\ref{th:enum-hardness-sms} would also hold for semantics based on any order~$\ll$ satisfying $\edgeset(w)\subsetneq \edgeset(w') \implies w \ll w'$.
\end{remark}

Finally, consider the problem \problemfont{Membership under \sms}, defined similarly to \problemfont{Membership under \mm}.
Although \problemfont{Membership under \sms} is clearly in \conp,
the constructions presented in this note does not seem adapted to show that it is \conp-hard.
However, we conjecture that it is.
\begin{conjecture}
    \problemfont{Membership Under \sms} is \conp-complete, for a fixed regular expression $R$.
\end{conjecture}

\section{References}

\ifdraft
    \nocite{*}
\fi
\printbibliography[heading=none]

@book{AroraBarak2009,
  author = {Arora, Sanjeev and Barak, Boaz},
  isbn = {978-0-521-42426-4},
  publisher = {Cambridge University Press},
  title = {Computational Complexity: A Modern Approach},
  year = 2009
}

@book{AbiteboulHullVianu1995, 
    author={Abiteboul, S. and Hull, Richard and Vianu, Victor}, 
    year={1995},
    title={Foundations of databases}, 
    ISBN={978-0-201-53771-0}, 
    publisher={Addison-Wesley}, 
    address={Reading, Mass}
}

@article{BaganBonifatiGroz2020,
  author    = {Guillaume Bagan and
               Angela Bonifati and
               Beno{\^{\i}}t Groz},
  title     = {A trichotomy for regular simple path queries on graphs},
  journal   = {J. Comput. Syst. Sci.},
  volume    = {108},
  pages     = {29--48},
  year      = {2020},
}

@book{BangJensenGutin2009, 
    author={Bang-Jensen, Jørgen and Gutin, Gregory Z.}, 
    year={2009},
    title={Digraphs : Theory, Algortihms and Applications},
    ISBN={978-0-85729-041-0}, 
    DOI={10.1007/978-1-84800-998-1}, 
    publisher={Springer London},  
    collection={Springer Monographs in Mathematics}, 
    address={London}, 
    series={Springer Monographs in Mathematics},    
}

@article{CruzMendelzonWood1987,
author = {Cruz, Isabel F. and Mendelzon, Alberto O. and Wood, Peter T.},
title = {A graphical query language supporting recursion},
year = {1987},
publisher = {Association for Computing Machinery},
address = {New York, NY, USA},
volume = {16},
number = {3},
issn = {0163-5808},
doi = {10.1145/38714.38749},
journal = {SIGMOD Rec.},
month = {dec},
pages = {323–330},
numpages = {8}
}

@inproceedings{DavidFrancisMarsault2023,
    author  = {Claire David and Nadime Francis and Victor Marsault},
    title   = {Run-Based Semantics for RPQs},
    pages   = {178--187},
    year    = {2023},
    month   = {8},
    booktitle = {Principles of Knowledge Representation and Reasoning (KR'23)},
    ee = {https://doi.org/10.24963/kr.2023/18},
}

@misc{DeutschEtAl2019,
  author    = {Alin Deutsch and
               Yu Xu and
               Mingxi Wu and
               Victor E. Lee},
  title     = {TigerGraph: {A} Native {MPP} Graph Database},
  year      = {2019},
  url       = {http://arxiv.org/abs/1901.08248},
  note = {Preprint \href{https://arxiv.org/abs/1901.08248}{arXiv:1901.08248}},
}

@inproceedings{DeutschEtAl2022,
      title = {Graph Pattern Matching in {GQL} and {SQL}/{PGQ}}, 
      author = {Alin Deutsch and 
                Nadime Francis and 
                Alastair Green and 
                Keith Hare and 
                Bei Li and 
                Leonid Libkin and 
                Tobias Lindaaker and 
                Victor Marsault and 
                Wim Martens and 
                Jan Michels and 
                Filip Murlak and 
                Stefan Plantikow and 
                Petra Selmer and 
                Hannes Voigt and 
                Oskar van Rest and 
                Domagoj Vrgo\v{c} and 
                Mingxi Wu and 
                Fred Zemke},
      year = {2022},
      booktitle = {SIGMOD'22},
      publisher = {{ACM}},
}

@misc{FrancisEtAl2018-arxiv,
  author =   {Nadime Francis and
              Alastair Green and
              Paolo Guagliardo and
              Leonid Libkin and
              Tobias Lindaaker and
              Victor Marsault and
              Stefan Plantikow and
              Mats Rydberg and
              Martin Schuster and
              Petra Selmer and
              Andr\'es Taylor},
  title = {Formal Semantics of the Language Cypher},
  year  = {2018},
  note = {Arxiv Preprint},
  url = {https://arxiv.org/abs/1802.09984}
}

@inproceedings{FrancisEtAl2018-sigmod,
  author = {Nadime Francis and Alastair Green and
            Paolo Guagliardo and Leonid Libkin and
            Tobias Lindaaker and Victor Marsault and
            Stefan Plantikow and Mats Rydberg and
            Petra Selmer and Andr\'es Taylor},
  title = {Cypher: An Evolving Query Language for Property Graphs},
  booktitle = {SIGMOD'18},
  year  = {2018},
  publisher = {ACM},
}

@inproceedings{FrancisEtAl2023-gpc,
      title = {GPC: A Pattern Calculus for Property Graphs},
      author = {Nadime Francis and
                Amélie Gheerbrant and
                Paolo Guagliardo and
                Leonid Libkin and
                Victor Marsault and
                Wim Martens and
                Filip Murlak and
                Liat Peterfreund and
                Alexandra Rogova and
                Domagoj Vrgoč},
      year = {2023},
      booktitle = {PODS'23},
      publisher = {ACM},
}

@inproceedings{FrancisEtAl2023-digest,
    title =	{{A Researcher’s Digest of GQL}},
    author = {Nadime Francis and 
              Am{\'e}lie Gheerbrant and
              Paolo Guagliardo and
              Leonid Libkin and
              Victor Marsault and
              Wim Martens and
              Filip Murlak and
              Liat Peterfreund and
              Alexandra Rogova and
              Domagoj Vrgo\v{c}
    },
    year = {2023},
    booktitle = {ICDT'23},
    series = {LIPIcs},
    ISBN =	{978-3-95977-270-9},
    ISSN =	{1868-8969},
    volume =	{255},
    editor =	{Geerts, Floris and Vandevoort, Brecht},
    publisher =	{Schloss Dagstuhl -- Leibniz-Zentrum f{\"u}r Informatik},
    URL =		{https://drops.dagstuhl.de/opus/volltexte/2023/17743},
    URN =		{urn:nbn:de:0030-drops-177434},
    doi =		{10.4230/LIPIcs.ICDT.2023.1},
}

@mastersthesis{Khichane2024,
    author = {Sara Khichane},
    title = {Study of New Semantics for Regular Path Queries},
    school = {Université Paris-Cité},
    year = {2024},
}

@inproceedings{ConsensMendelzon1990,
author = {Consens, Mariano P. and Mendelzon, Alberto O.},
title = {GraphLog: a visual formalism for real life recursion},
year = {1990},
isbn = {0897913523},
publisher = {Association for Computing Machinery},
address = {New York, NY, USA},
url = {https://doi.org/10.1145/298514.298591},
doi = {10.1145/298514.298591},
booktitle = {Proceedings of the Ninth ACM SIGACT-SIGMOD-SIGART Symposium on Principles of Database Systems},
pages = {404–416},
numpages = {13},
location = {Nashville, Tennessee, USA},
series = {PODS '90}
}

@inproceedings{MartensNiewerthTrautner2020,
  author    = {Wim Martens and
               Matthias Niewerth and
               Tina Trautner},
  editor    = {Christophe Paul and
               Markus Bl{\"{a}}ser},
  title     = {A Trichotomy for Regular Trail Queries},
  booktitle = {STACS'20},
  series    = {LIPIcs},
  volume    = {154},
  pages     = {7:1--7:16},
  publisher = {Schloss Dagstuhl - Leibniz-Zentrum f{\"{u}}r Informatik},
  year      = {2020},
  url       = {https://doi.org/10.4230/LIPIcs.STACS.2020.7},
  doi       = {10.4230/LIPIcs.STACS.2020.7},
}

@misc{MarsaultMeyer2024,
  author    = {Victor Marsault and Antoine Meyer},
  title     = {Framework for comparing walk-outputting RPQs semantics},
  note      = {In preparation},
  year = {2024},
}

@article{MendelzonWood1995,
  author    = {Alberto O. Mendelzon and
               Peter T. Wood},
  title     = {Finding Regular Simple Paths in Graph Databases},
  journal   = {{SIAM} J. Comput.},
  volume    = {24},
  number    = {6},
  pages     = {1235--1258},
  year      = {1995}
}

@manual{Strozecki2021,
    author = {Strozecki, Yann},
    title = {Enumeration Complexity: Incremental Time, Delay and Space},
    organization = {Universit\'e de Versailles -- Saint-Quentin-en-Yvelines},
    year = {2021},
    note =  {Habilitation thesis}
}

@misc{GQL-ISO,
    shorthand = {GQL},
    author = {{International Organization for Standardization}},
    title = {{GQL}},
    howpublished = {\url{https://www.iso.org/standard/76120.html}},
    note = {Standard under development ISO/IEC CD 39075},
    year={2024}
}

@misc{SQLPGQ-ISO,
    shorthand = {PGQ},
    author = {{International Organization for Standardization}},
    title = {SQL — Part 16: SQL Property Graph Queries (SQL/PGQ)},
    howpublished = {\url{https://www.iso.org/standard/79473.html}},
    note = {Standard ISO/IEC 9075-16:2023},
    year = {2023},
}

@misc{PGQL1.4,
  shorthand = {PGQL},
  author    = {Property Graph Query Language},
  title     = {{PGQL 1.4 Specification}},
  url       = {https://pgql-lang.org/spec/1.4/},
  year      = {2021}
}

@misc{SPARQL1.1PP,
    shorthand = {SPARQLPP},
    title = {{SPARQL} 1.1 Query Language, Section 9: Property paths},
    howpublished = {\url{https://www.w3.org/TR/sparql11-query/\#propertypaths}},
    author = {{World Wide Web Consortium}},
    year = {2013},
}

@Control{biblatex-control,
  options = {3.8:0:0:1:0:1:1:0:1:0:0:2:3:2:100:1:0:0:3:2:79:${}^{+}$:+:nyt},
}

\end{document}